\documentclass[reqno,12pt]{amsart}
\usepackage{hyperref}
\usepackage{amssymb}
\usepackage{amsmath}
\usepackage{amsthm}
\usepackage{mathrsfs}
\usepackage[arrow,matrix,curve]{xy}
\usepackage{graphicx}
\usepackage{hyperref}
\usepackage{amscd,verbatim}
\usepackage{mathtools}
\usepackage{comment}

\newcommand\dirac{\slash\!\!\!\partial}

\usepackage{fullpage,enumerate}
\newcommand{\bea}{\begin{eqnarray}}
\newcommand{\eea}{\end{eqnarray}}
\newcommand{\beq}{\begin{equation}}
\newcommand{\eeq}{\end{equation}}

\newcommand\A{{\bf A}}
\newcommand\B{{\bf B}}

\newtheorem{theorem}{Theorem}[section]

\newtheorem{proposition}[theorem]{Proposition}

\theoremstyle{remark}

\newtheorem{example}[theorem]{Example}

\numberwithin{equation}{section}
\numberwithin{theorem}{section}

\newcommand\cE{\mathcal{E}}
\newcommand\cF{\mathcal{F}}

\newcommand\cK{\mathcal{K}}
\newcommand\cM{\mathcal{M}}

\newcommand\cL{\mathcal{L}}

\newcommand\cP{\mathcal{P}}

\newcommand\cS{\mathcal{S}}
\newcommand\cV{\mathcal{V}}

\newcommand{\Z}{\ensuremath{\mathbb Z}}
\newcommand{\R}{\ensuremath{\mathbb R}}

\newcommand{\CC}{{\mathbb C}}

\newcommand{\QQ}{{\mathbb Q}}
\newcommand{\RR}{{\mathbb R}}
\newcommand{\TT}{{\mathbb T}}
\newcommand{\ZZ}{{\mathbb Z}}

\begin{document}

\title[T-duality simplifies bulk-boundary correspondence]
{T-duality simplifies bulk-boundary correspondence}

\author[V Mathai]{Varghese Mathai}

\address[Varghese Mathai]{
Department of Pure Mathematics,
School of  Mathematical Sciences, 
University of Adelaide, 
Adelaide, SA 5005, 
Australia}

\email{mathai.varghese@adelaide.edu.au}

\author[G.C.Thiang]{Guo Chuan Thiang}

\address[Guo Chuan Thiang]{
Department of Pure Mathematics,
School of  Mathematical Sciences, 
University of Adelaide, 
Adelaide, SA 5005, 
Australia}

\email{guo.thiang@adelaide.edu.au}

\begin{abstract}
Recently we introduced T-duality in the study of topological insulators.
In this paper, we study the bulk-boundary correspondence for three phenomena in condensed matter physics, 
namely, the quantum Hall effect, the Chern insulator, and time reversal invariant topological insulators. 
In all of these cases, we show that T-duality trivializes the bulk-boundary correspondence.
\end{abstract}

\maketitle


\section{Introduction}
\label{sec:intro}
The most interesting and potentially useful feature of a topological insulator is the presence of gapless modes on its boundary, despite it being an insulating material in the bulk. The bulk-boundary correspondence describes a correlation between certain topological features of the bulk description and those of boundary phenomena. Generically, gapless boundary modes are expected to arise at the interface of two samples whose bulk Hamiltonians are characterized by different topological invariants, reflecting the topological obstruction in continuously deforming one Hamiltonian into the other. A stronger version of the correspondence states that there is a homomorphism mapping a bulk topological invariant to a boundary topological invariant. Such a statement not only implies that the gapless boundary modes are topologically protected, but also that bulk and boundary quantities are  \emph{simultaneously} quantized.

In the context of the 2D integer quantum Hall effect \cite{Bellissard}, there is a correspondence between the quantized Hall conductivity and chiral edge currents. This phenomenon was explained in \cite{Kellendonk1,Kellendonk2,Kellendonk3} using the language of $K$-theory and Connes' noncommutative geometry \cite{Connes94}, as well as in \cite{Hatsugai,Elbau,Avila,Kotani,Bourne} using mathematical techniques of various sophistication and generality. With the recent experimental discovery of Chern insulators \cite{CZ,JM} and topological insulators protected by time-reversal symmetry \cite{KW,Hsieh}, it is desirable to have a mathematical framework in which a general bulk-boundary correspondence principle can be stated. 

In this paper, we formulate the bulk-boundary correspondence as a topological boundary map associated to an extension of a bulk algebra of observables by a boundary algebra, in the sense of \cite{Kellendonk1,Kellendonk4}. More precisely, we assume that the bulk $C^*$-algebra $\mathscr{C}$ arises as a crossed product of the codimension-1 boundary algebra $\mathscr{J}$ by $\ZZ$, and relate them through the associated Toeplitz extension. The Fermi projection for a bulk Hamiltonian defines an element of $K_0(\mathscr{C})$, which has a homomorphic image in $K_1(\mathscr{J})$ under the boundary map in $K$-theory. We remark that the role of $K$-theory in the systematic study of topological phases with symmetry has already been discussed in \cite{AK,FM,T,T2,MT}.

In \cite{MT}, we introduced T-duality as a tool to study topological insulators. As an example, we explained how the Chern number and the rank of valence bundles for 2D Chern insulators are interchanged under the T-duality transformation. We will proceed to show that the bulk-boundary homomorphism becomes trivialized under T-duality. This turns out to be a fairly generic phenomenon --- a relatively complicated bulk-boundary map becomes a trivial map when viewed on the T-dual side. We develop noncommutative, twisted, and real versions of T-duality, and demonstrate the principle for the quantum Hall effect, the Chern insulator, and the time-reversal invariant topological insulator. In the latter, Real $KR$-theory or Quaternionic $KQ$-theory groups are required, but T-duality simplifies the mathematical description by translating them to ordinary real $KO$-theory groups. The Fu--Kane--Mele (FKM) $\ZZ_2$-invariant \cite{KM2,FK2,FKM} discovered in the physics literature (which as explained in \cite{deNittis} is really the FKMM invariant \cite{FKMM} introduced some years earlier) can then be understood as the T-dual to the classical Stiefel--Whitney classes. These results illustrate the potential utility of T-duality in the field of topological insulators. \bigskip

\begin{table}[h]
\begin{tabular}{|l|l|l|}
  \hline
  Topological invariant & T-dual & Group\\
  \hline
  \hline
  \multicolumn{3}{|c|}{$d=2$, No $\mathsf{T}$}\\
  \hline
  $c_1$ & rank & $\ZZ$\\ 
  rank & $c_1$ & $\ZZ$\\
  \hline
  \multicolumn{3}{|c|}{$d=2$, $\mathsf{T}^2=-1$}\\
  \hline
  FKM & $w_2$ & $\ZZ_2$\\
  $Q$-rank & $p_1$ & $\ZZ$ \\
  \hline

\hline
  \multicolumn{3}{|c|}{$d=3$, $\mathsf{T}^2=-1$}\\
  \hline
  $\vartheta$ (strong FKM) & $w_1$ & $\ZZ_2$\\
  weak FKM & $w_2$ & $3\ZZ_2$\\
  $Q$-rank & $p_1$ & $\ZZ$ \\
\hline
\hline
  \multicolumn{3}{|c|}{Bulk-boundary homomorphism}\\
  \hline
  $\partial$ & $\iota^*$ & {N/A}\\
\hline
\end{tabular}\bigskip
\caption{A dictionary for translating between topological invariants for topological insulators and their T-duals.}
\label{table:dictionary}
\end{table}

\tableofcontents

\section{The relevant Hamiltonians }
\subsection{The quantum Hall effect}\label{section:IQHE}
We begin by reviewing the construction of the Hamiltonian in the integer quantum Hall effect (IQHE).
Consider Euclidean space $\RR^2$ equipped with its usual
metric $(dx^2+dy^2)$, and symplectic area form
$\omega  = dx\wedge dy$.
The Euclidean group $G=\RR^2 \rtimes SO(2)$ acts transitively on $\RR^2$ by affine transformations
The torus $\TT^2$ can be realised as the
quotient of
$\RR^2$ by the action of its fundamental group $\ZZ^2$.

Let us now pick an  electromagnetic vector potential, that is, a 
1-form $\A$ such that $d\A = \B= \theta\omega$ is the magnetic field in some suitable units, for some
fixed $\theta \in (0,1)$. 
As in geometric quantization we may regard $\A$ as defining a connection
$\nabla = d+i\A$ on a line bundle $\cL$ over $\RR^2$, whose curvature is
$i\B=i\theta\omega$.
Using the Riemannian metric  the Hamiltonian of an
electron in this field is given in suitable units by
$$H = H_{\A} = \frac 12\nabla^*\nabla = \frac 12(d+i\A)^*(d+i\A).$$
The underlying $G$-invariance of the theory is clear.
In a real material this Hamiltonian would be modified by the addition of a real-valued
potential $V$.
By taking $V$ to be invariant under $\ZZ^2$, that is periodic, this perturbation is
given a crystalline type structure.
The spectrum of the unperturbed Hamiltonian $H_\A$
for $\A = -\theta ydx$ has been computed by Landau in the 1930s, cf. \cite{Avron}.
It has only discrete eigenvalues with infinite multiplicity. 
Any $\A$ is cohomologous to $-\theta ydx$ since they both have
$\theta\omega$ as differential, and forms differing by an exact form $d\phi$
 give
equivalent models: in fact, multiplying the wave functions by
$\exp(i\phi)$ shows that the models for $\A$ and $-\theta ydx$ are
unitarily equivalent.
This equivalence also intertwines the $\ZZ^2$-actions so that the spectral
densities for the two models also coincide.
However, the perturbed Hamiltonian $H_{\A,V} = H_\A +V$, which
is the key to the quantum Hall effect,  has unknown spectrum
for general $\ZZ^2$-invariant $V$. For $\gamma = (m, n) \in \ZZ^2$,
let $\psi_\gamma(x, y) = n\theta x$ be a function on $\RR^2$. 
It satisfies $\gamma^*\A=\A = d \psi_\gamma$. Also 
$\psi_\gamma(0, 0)=0$ for all $\gamma  \in \ZZ^2$ and 
$\psi_\gamma(\gamma') = \theta m'n$ where $\gamma' = (m', n') \in \ZZ^2$.
Define a projective unitary action $T^\sigma$ of $\ZZ^2$ on $L^2(\RR^2)$ as follows:
\begin{align*}
U_\gamma(f)(x, y) & = f(x-m, y-n)\\
S_\gamma(f)(x,y) & = \exp(-2\pi i \psi_\gamma(x,y)) f(x,y)\\
T_\gamma^\sigma &=U_\gamma\circ S_\gamma.
\end{align*}
Then the operators $T_\gamma^\sigma$, also known as {\em magnetic translations}, satisfy $T^\sigma_e={\rm{Id}}, \,\, T^\sigma_{\gamma_1}
T^\sigma_{\gamma_2} = \sigma(\gamma_1, \gamma_2)T^\sigma_{\gamma_1\gamma_2}$,
where $\sigma(\gamma, \gamma') = \exp(-2\pi i\theta m'n)$, which is a {\em multiplier} (or {\em 2-cocycle})
on $\ZZ^2$, that is, it satisfies,
\begin{enumerate}
\item $\sigma(\gamma, e)=\sigma(e, \gamma)=1$ for all $\gamma \in \ZZ^2$;
\item $\sigma(\gamma_1, \gamma_2)\sigma(\gamma_1\gamma_2, \gamma_3)=\sigma(\gamma_1, \gamma_2\gamma_3)\sigma(\gamma_2, \gamma_3)$
for all $\gamma_j \in \ZZ^2,\, j=1,2,3.$
\end{enumerate}
An easy calculation shows that $T^\sigma_\gamma \nabla = \nabla T^\sigma_\gamma$ and taking adjoints, 
$T^\sigma_\gamma \nabla^* = \nabla^* T^\sigma_\gamma$. Therefore $T^\sigma_\gamma H_\A = H_\A T^\sigma_\gamma$.
Also, since $V$ is periodic, $T^\sigma_\gamma V = V T^\sigma_\gamma$. We conclude that for all $\gamma\in \ZZ^2$,
$T^\sigma_\gamma H_{\A,V} = H_{\A, V} T^\sigma_\gamma$, that is, the Hamiltonian commutes with magnetic translations.
The commutant of the projective action $T^\sigma$ is the projective action $T^{\bar\sigma}$. If $\lambda$ lies in a spectral gap
of $H_{\A, V} $, then the Riesz projection is $p_\lambda(H_{\A, V})$ where $p_\lambda$ is a  smooth compactly supported
function which is identically equal to 1 in the interval $[{\rm inf} H_{\A, V}, \lambda]$, 
and whose support is contained in the interval $[-\varepsilon +{\rm inf} H_{\A, V}, \lambda+ \varepsilon]$ for some $\varepsilon>0$ small enough.
Then  $$p_\lambda(H_{\A, V})\in C^*(\ZZ^2, \bar\sigma)\otimes\cK(L^2(\cF)),$$ where $\cF$ is a fundamental
domain for the action of $\Z^2$ on $\R^2$, and $p_\lambda(H_{\A, V})$ defines an element in $K_0( C^*(\ZZ^2, \bar\sigma))$. By the {\em gap hypothesis}, the Fermi level of  $H_{\A, V}$ lies in a
spectral gap.
We also study the case when 
$V$ is aperiodic with hull a Cantor set $\Sigma$. In this case, the Riesz projection to a spectral gap, or the {\em Fermi projection}, defines an element 
in $K_0( C(\Sigma) \rtimes_{\bar\sigma}\ZZ^2))$.

\subsection{The Chern insulator}\label{section:cherninsulator}
In \cite{Hal}, Haldane introduced an example of what is now called a \emph{Chern insulator} where instead of magnetic translational symmetry, the Hamiltonian has ordinary $\ZZ^2$ translation invariance, and ordinary Bloch--Floquet theory \cite{RS} suffices for its analysis. Such translation-invariant Hamiltonians lead to the study of their associated vector bundles (Bloch bundles) over the Brillouin 2-torus $\TT^2$, and for insulators the topological invariants of the sub-bundle of eigenstates (comprising the valence bands) lying below the Fermi level is of particular interest. For the abstract arguments that follow, the Hamiltonians are formulated directly on quasi-momentum space as a continuous family of Bloch Hamiltonians over $\TT^2$, rather than on real space.

The basic model of a Chern insulator is a two-band Hamiltonian in 2D, which up to an overall shift in the energy level has the generic form \cite{SPF}
\begin{equation*}
    H(k)= \mathbf{h}(k)\cdot\vec{\sigma}\equiv \sum_{i=1}^3 h_i(k)\sigma_i, \;\;k\in\TT^2,
\end{equation*}
where $\sigma_i$ are the Pauli matrices and $\mathbf{h}=(h_1,h_2,h_3)$ is a smooth map from the Brillouin torus $\TT^2$ to $\mathbb{R}^3$. Since $H(k)^2=|\mathbf{h}(k)|^2$, the gapped condition for an insulator is ensured if $\mathbf{h}$ is nowhere zero. Assuming this, we can form the ``spectrally-flattened'' Hamiltonian $\widehat{H}(k)=\widehat{\mathbf{h}}(k)\cdot\vec{\sigma}$ where $\widehat{\mathbf{h}}=\mathbf{h}/|\mathbf{h}|$, and the valence band is given by the projection onto its negative eigenbundle $P(k)=\frac{1}{2}(1-\widehat{H}(k))$. In particular, the valence band is a line bundle $\mathcal{E}$ and has the first Chern number $c(\mathcal{E})$ as a topological invariant:
\begin{equation*}
c(\mathcal{E})=\frac{1}{4\pi}\int_{\TT^2}\mathrm{d}k\,\widehat{\mathbf{h}}(k)\cdot\frac{\partial\widehat{\mathbf{h}}(k)}{\partial{k_1}}\wedge\frac{\partial\widehat{\mathbf{h}}(k)}{\partial{k_2}}.
\end{equation*}

In general, vector bundles $\mathcal{E}$ over $\TT^2$ can be characterized by their Chern character $\mathrm{Ch}(\mathcal{E})=\mathrm{rank}(\mathcal{E})+c_1(\mathcal{E})$, where $c_1(\mathcal{E})$ is the first Chern class. Then the (first) Chern number $c(\mathcal{E})$ is obtained from pairing $c_1(\mathcal{E})$ with the fundamental class of $\TT^2$. The Chern number is determined by the determinant line bundle, so to construct model Hamiltonians with valence bundles having arbitrary rank and Chern number, it suffices to consider line bundles in the above two-band models with the required Chern number, then augment by some trivial bundles to make up the rank. Such line bundles can be obtained by a judicious choice of the function $\mathbf{h}$ \cite{SPF}. For instance, a family of models which realize valence bands with Chern numbers $0,\pm 1$ or $\pm 2$ is
\begin{equation}
\mathbf{h}(k)=(\cos k_1, \cos k_2, m+a\cos (k_1+k_2)+b(\sin k_1+\
\sin k_2)),\label{chernhamiltonian}
\end{equation}
where $m,a,b$ are real parameters. The Chern number of the valence band $\mathcal{E}$ can be computed to be \cite{SPF}
\begin{equation*}
    c(\mathcal{E})=\mathrm{sgn}(-m-a)+\frac{1}{2}[\mathrm{sgn}(m-a+2b)+\mathrm{sgn}(m-a-2b)].
\end{equation*}

\subsection{The time-reversal invariant topological insulator}
The defining feature of a time-reversal invariant (or $\mathsf{T}$-invariant) Hamiltonian formulated on the Brillouin torus $\TT^d$ is the presence of an antiunitary anti-involutary time-reversal operation $\mathsf{T}$ on the Bloch bundle, satisfying the following commutation relation for the Bloch Hamiltonians
\beq
    H(k)=\mathsf{T}H(-k)\mathsf{T}^{-1}, \qquad k\in\TT^d.\nonumber
\eeq
In other words, $\mathsf{T}$ furnishes the Bloch bundle with a Quaternionic structure which the Bloch Hamiltonians must be compatible with. The extra involution $\varsigma:k\mapsto -k$ arises due to the Fourier transform of ordinary complex conjugation on real space. The requirement that $\mathsf{T}^2=-1$ means that each Bloch eigenstate $v$ has a Kramers partner $\mathsf{T}v$; in particular, there is a degeneracy at the fixed points of $\varsigma$.

A standard way to construct $\mathsf{T}$-invariant Bloch Hamiltonians is to begin with an unconstrained one, such as the $H(k)$ from the Chern insulator example, and take $\widetilde{H}(k)=H(k)\oplus \overline{H}(-k)$ where $\overline{H}(k)$ denotes the complex conjugate of $H(k)$. Then we can take 
\beq
\mathsf{T}=\begin{pmatrix} 0 & 1 \\ -1 & 0 \end{pmatrix}\mathsf{K},\label{TIhamiltonian}
\eeq
where $\mathsf{K}$ is complex conjugation, and verify that $\mathsf{T}\widetilde{H}(-k)\mathsf{T}^{-1}=\widetilde{H}(k)$.

In $d=2$, an explicit model originating from a tight-binding Hamiltonian on a honeycomb lattice was given by Kane--Mele \cite{KM,KM2}. Another example was constructed by Bernevig--Hughes--Zhang \cite{BHZ} starting from a square lattice Hamiltonian, and their Bloch Hamiltonians are of the form  \eqref{TIhamiltonian} with
\beq
H(k)=[\Delta +\cos k_1 + \cos k_2]\sigma_3 + A(\sin k_1\sigma_1+\sin k_2\sigma_2),\nonumber
\eeq
where $\Delta$ and $A$ are some real parameters.

The particular model used to obtain $\mathsf{T}$-invariant Hamiltonians may have some additional symmetries, such as point symmetries or conserved spin in the $3$-direction perpendicular to the plane of motion. These extra symmetries may either be genuine physical constraints, or accidental features of the model. In general, it is the totality of the symmetry constraints which determine the relevant category of bundles and therefore the relevant topological invariants to use. For example, if conservation of spin in the $3$-direction is required, one can define the so-called spin-Chern numbers \cite{Sheng,Prodan2}, complementing the various $\ZZ_2$-valued indices of Fu--Kane--Mele \cite{KM2,FK2,FKM}; see \cite{Avila} for further discussion. When spatial inversion $\mathsf{P}$ and time-reversal $\mathsf{T}$ are separately symmetries of the Hamiltonian, there is a lift of the $\ZZ_2$-invariant to an integer one \cite{FM}. When $\mathsf{P}\mathsf{T}$ is a symmetry but $\mathsf{P}$ and $\mathsf{T}$ are not separately symmetries, the $\ZZ_2$ invariant disappears \cite{MT}.

We are only concerned with the $\ZZ_2$-invariant associated with 2D $\mathsf{T}$-invariant topological insulators, and its correspondence to a boundary $\ZZ_2$-invariant, so only the constraint of $\mathsf{T}$-symmetry (along with the implicit $\ZZ^2$-translational symmetry) applies. Mathematically, this means that the valence bundles for $\mathsf{T}$-invariant topological insulators are to be thought of as Quaternionic bundles over the 2-torus $\widehat{\TT}^2$ with involution $\varsigma$, see Section \ref{section:quaternionicbundles}. Similarly, we use Quaternionic bundles over $\widehat{\TT}^3$ to model 3D $\mathsf{T}$-invariant topological insulators \cite{Hsieh}, for which there are four $\ZZ_2$-invariants \cite{FKM}.

\section{Bulk-boundary homomorphism and Toeplitz extension}
An important physical feature of the integer quantum Hall effect is the correspondence between the transverse Hall conductivity in the bulk and chiral edge currents along the boundary. There is a useful heuristic semi-classical picture for this correspondence --- the cyclotron orbits of electrons, under the application of a longitudinal electric field and a constant magnetic field perpendicular to the plane of motion, are intercepted by the boundary of a finite sample, giving rise to chiral currents along the boundary.

Rigorous work establishing this bulk-boundary correspondence for the IQHE was carried out in \cite{Kellendonk1,Kellendonk2,Kellendonk3}, for both discrete and continuous models (see also \cite{Bourne,Elbau}). The basic idea was to extend the $C^*$-algebra associated to bulk observables by one pertaining to boundary observables, through a short exact sequence, whose middle algebra describes a family of bulk-with-boundary systems. This extension turns out to be a Toeplitz-like extension for the crossed product by a complementary group of translational symmetries $\ZZ$ transverse to the boundary (a related Wiener--Hopf extension was used for continuous $\mathbb{R}$ symmetry). Its corresponding $K$-theory (cyclic) long exact sequence is the Pimsner--Voiculescu exact sequence, and the $K$-theory boundary map taking $K_0$ of the bulk algebra to $K_1$ of the boundary algebra, together with a dual boundary map in cyclic cohomology (see Section \ref{section:dualPV}), constituted the bulk-boundary correspondence. Roughly speaking, this correspondence maps the noncommutative Chern number for the Fermi projection to a noncommutative winding number. We extract the main idea behind their construction which we use to model a general notion of bulk-boundary correspondence in $K$-theory.

Suppose we have a $C^*$-algebra of observables $\mathscr{C}$ associated to a bulk physical system, and another $C^*$-algebra $\mathscr{J}$ associated to a codimension-1 boundary. We assume that there is $\ZZ$ translation invariance (or more generally covariance with action $\alpha$) in the spatial dimension transverse to the boundary, which translates to $\mathscr{C}$ being a crossed product $\mathscr{C}=\mathscr{J}\rtimes_\alpha\ZZ$. There is a natural Toeplitz extension associated to such a crossed product (\cite{Pimsner}, 10.2 of \cite{Blackadar}),
\beq
    0\longrightarrow\left(\mathscr{J}\otimes C_0(\ZZ)\right)\rtimes_{\alpha\otimes\tau}\ZZ\longrightarrow \left(\mathscr{J}\otimes C_0(\ZZ\cup\{\infty\})\right)\rtimes_{\alpha\otimes\tau}\ZZ\xrightarrow{\mathrm{ev}_\infty} \mathscr{J}\rtimes_\alpha\ZZ\longrightarrow 0,\label{toeplitzextension}
\eeq
where $\tau$ is left translation on $\ZZ$ and fixes $\infty$. Abbreviating the middle algebra by $\mathscr{E}$ (for ``extension'') and using $C_0(\ZZ)\rtimes_\tau\ZZ\cong\mathcal{K}(l^2(\ZZ))$, we can rewrite this as
\beq
    0\longrightarrow\mathscr{J}\otimes\mathcal{K}\longrightarrow \mathscr{E}\xrightarrow{\mathrm{ev}_\infty}\mathscr{C}\longrightarrow 0.\label{toeplitzextensionshort}
\eeq
In \eqref{toeplitzextension}, the $\ZZ$ in $C_0(\ZZ)$ can be interpreted as the (classical) variable labelling the position of the boundary, relative to some arbitrary choice of origin, so the $\alpha\otimes\tau$ action of $\ZZ$ encodes this covariance. We allow the position of the boundary to be ``at infinity'' in the middle algebra $\mathscr{E}$, which now has the interpretation of a continuous family of bulk-with-boundary algebras such that evaluation at infinity gives the algebra $\mathscr{C}$ for the bulk-without-boundary. Thus, the manner in which the bulk-without-boundary description is topologically linked with the bulk-with-boundary description is encoded in the short exact sequence \eqref{toeplitzextensionshort}. The boundary map $\partial:K_0(\mathscr{C})\rightarrow K_1(\mathscr{J}\otimes\mathcal{K})\cong K_1(\mathscr{J})$ is taken to be the ($K$-theoretic) bulk-boundary homomorphism in this paper.

\begin{example}
The Fermi projection $p_\lambda$ for the Hamiltonian in the IQHE defines an element of $K_0(\mathscr{C})$ where $\mathscr{C}=C^*(\ZZ^2,\bar{\sigma})$ as explained in Section \ref{section:IQHE}. In the presence of disorder, a more complicated crossed product needs to be used \cite{Bellissard}. In the detailed tight-binding model in \cite{Kellendonk1,Kellendonk2}, the bulk algebra $\mathscr{C}$ describes operators which are homogeneous in the plane, $\mathscr{J}\otimes\mathcal{K}$ describes operators which are homogeneous along the boundary but compact in the transverse direction, while $\mathscr{E}$ describes homogeneous half-plane operators with compact boundary conditions. The Hall conductivity is obtained by pairing the $K_0$ class of the Fermi projection with the conductivity cyclic 2-cocycle, following \cite{Bellissard}. The dual boundary map in cyclic cohomology takes this 2-cocycle to a 1-cocycle, whose pairing with $\partial[p_\lambda]$ equals the original pairing. Furthermore, the latter pairing was explained to be the edge current carried by the edge states for the Hamiltonian restricted to the half-plane. A continuous analogue of this work appeared in \cite{Kellendonk3}, and is summarized in \cite{Kellendonk4}.
\end{example}

\begin{example}
For a band insulator with $\ZZ^d$ translational symmetry, there is a family of Bloch Hamiltonians over the the Brillouin torus $\TT^d$, and the valence subbundle over the Brillouin torus $\TT^d$ defines a projection in $C(\TT^d)\cong C(\TT^{d-1})\rtimes_\mathrm{id}\ZZ$. The simplest example is that of the 2D Chern insulator described in Section \ref{section:cherninsulator}, which can be thought of as a commutative version of the IQHE without magnetic fields and Landau levels. Typically, $\TT^{d-1}$ is a sub-Brillouin torus coming from the remaining $\ZZ^{d-1}$ symmetry for the boundary. This need not be the case, and we can think of some of the torus dimensions as arising from angular variables for a parameter space. This is the paradigm of ``virtual'' topological insulators, which are discussed in \cite{Prodan}.

In the presence of further symmetry constraints such as that of time-reversal, the projection onto the valence bundle defines a class in a refined version of $K$-theory, such as $KR$, $KQ$, or their equivariant versions. In these cases, $\mathscr{J}$ is regarded as a real $C^*$-algebra, the crossed product and Toeplitz extension are real, and we have to keep track of the sign of the $K$-theory degree, e.g.\ $\partial:KO_0(\mathscr{C})\rightarrow KO_{-1}(\mathscr{J})\cong KO_7(\mathscr{J})$, see Section \ref{section:quaternionicbundles}.
\end{example}

It should be noted that the $K$-theory groups only provide primitive topological invariants, which should then be paired with some element of a suitable dual theory in order to extract numerical invariants with physical meaning. For example, a generic valence vector bundle for a 2D Chern insulator is characterized by its rank and first Chern class. Its pairing with the fundamental class extracts the first Chern number, and the explicit formula for this pairing reveals the connection between the first Chern number and the Hall conductivity; pairing with a 0-trace extracts the rank (or the gap-labelling group in the aperiodic case). As an alternative to cyclic cohomology, Kasparov theory can be used to relate pairings between $K$-homology and $K$-theory \cite{Bourne} in the case of the IQHE, and can potentially apply to the time-reversal invariant insulators for which real $K$-theory is required. We are mainly concerned with studying the bulk-boundary correspondence at the level of these primitive $K$-theoretic invariants.

There are several other complementary approaches to the bulk-boundary correspondence \cite{Hatsugai,Graf,Avila} which do not use topological boundary maps but analyze directly the relationship between a bulk Hamiltonian on the plane and its restriction to some subspace (e.g.\ a half-plane) of the plane with boundary conditions applied (see also \cite{Loring} for a $K$-theoretic approach). Conceptually, these two Hamiltonians describe different physical setups, whereas the correspondence established in \cite{Kellendonk1,Kellendonk2,Kellendonk3} is a stronger and manifestly topological \emph{simultaneous} quantization of bulk and boundary phenomena. The language of $C^*$-algebra extensions precisely facilitates a concurrent description of bulk-with-boundary observables and those of the bulk without boundary within a \emph{single} extension algebra $\mathscr{E}$.

\subsubsection*{Pimsner--Voiculescu boundary map}
A breakthrough in the computation of the $K$-theory of crossed products by an action of $\ZZ$ was made in \cite{Pimsner}. The authors used a Toeplitz extension \eqref{toeplitzextension} and identified the $K$-theory of the middle term with the $K$-theory of $\mathscr{J}$. Then the associated cyclic exact sequence becomes the Pimsner--Voiculescu (PV) exact sequence
\beq
    \xymatrix{ K_0(\mathscr{J}) \ar[r]^{1-\alpha_*} & K_0(\mathscr{J}) \ar[r]^{j_*\;\;\;\;} & K_0(\mathscr{J}\rtimes_\alpha \ZZ) \ar[d]^\partial & \\
  K_1(\mathscr{J}\rtimes_\alpha \ZZ)  \ar[u] & K_1(\mathscr{J}) \ar[l]^{\;\;\;\;j_*} & K_1(\mathscr{J}) \ar[l]^{\;1-\alpha_*}  }\label{PVexactsequence}
\eeq
where $\alpha_*$ is the induced map in $K$-theory under the automorphism $\alpha(1)$ of $\mathscr{J}$. There is a real version of this result which is 24-cyclic \cite{Rosenberg2,Schroder}. We will use the PV exact sequence extensively for the computation of the bulk-boundary homomorphism $\partial$.

\subsection{Bulk-boundary for twisted actions of $\ZZ^2$}
In physical examples, the bulk algebra $\mathscr{C}$ often arises as an iterated crossed product by $\ZZ$. In other words, it is of the form $\mathscr{J}\rtimes_\alpha\ZZ$ where $\mathscr{J}$ is itself a crossed product by $\ZZ$.

Let $\mathcal{A}$ be a complex unital $C^*$-algebra. For a twisted crossed product $\mathcal{A}\rtimes_\sigma\ZZ^2\equiv\mathcal{A}\rtimes_{(\alpha,\sigma)}\ZZ^2$, with $\sigma$ a $\mathrm{U}(1)$-valued 2-cocycle, the Packer--Raeburn decomposition theorem \cite{PR} allows us to rewrite it as an iterated untwisted crossed product \mbox{$(\mathcal{A}\rtimes_{\alpha_1}\ZZ^{(1)})\rtimes_{\alpha_2}\ZZ^{(2)}$}. In the second crossed product, the $\alpha_2$-action on $\mathcal{A}\rtimes_{\alpha_1}\ZZ^{(1)}$ comprises the original $\alpha$ action of $\ZZ^{(2)}$ on $\mathcal{A}$ together with multiplication by a phase $e^{2\pi i\theta} $ on $\delta_n\in\mathcal{A}\rtimes_{\alpha_1}\ZZ^{(1)}, n\in\ZZ$ (see the formulae in Theorem 4.1 of \cite{PR}), where $\theta$ is an angle determined by $\sigma$ as in Section \ref{section:IQHE}.

Note that we can deform the automorphism $\alpha_2$ (thus also $\alpha$) into one where $\theta=0$. For example, when $\mathcal{A}=\mathbb{C}$, we have the noncommutative torus $A_\theta$, in which $\alpha_2(1)$ is rotation of $\mathbb{C}\rtimes_{\alpha_1}\ZZ^{(1)}\cong C(S^1)$ by $2\pi\theta$. The rotation angle can be continuously decreased to zero, upon which $\alpha_2$ becomes trivial.

The unital inclusion $\mathbb{C}\xhookrightarrow{\iota}\mathcal{A}$ is $\ZZ^{(1)}$-equivariant (the $\ZZ^{(1)}$-action restricts to the trivial action on $\mathbb{C}$), and induces a homomorphism (also called $\iota$) on the crossed product\footnote{We sometimes write $\mathcal{A}\rtimes G$ instead of $\mathcal{A}\rtimes_\alpha G$ when the group action is understood.} $C(S^1)\cong\mathbb{C}\rtimes\ZZ^{(1)}\xhookrightarrow{\iota}\mathcal{A}\rtimes_{\alpha_1}\ZZ^{(1)}$ (Corollary 2.4.8 of \cite{Williams}), and therefore a homomorphism $\iota_*$ on their $K$-theory groups (which is not necessarily injective). Note that $\mathbb{C}\rtimes\ZZ^{(1)}\xhookrightarrow{\iota}\mathcal{A}\rtimes_{\alpha_1}\ZZ^{(1)}$ is itself $\ZZ^{(2)}$-equivariant (the $\ZZ^{(2)}$ action restricts to $\iota(\mathbb{C}\rtimes\ZZ^{(1)})$), and induces a homomorphism (also called $\iota$) on the iterated crossed product \mbox{$A_\theta\cong\mathbb{C}\rtimes_\sigma\ZZ^2\xhookrightarrow{\iota}\mathcal{A}\rtimes_\sigma\ZZ^2$}, as well as their $K$-theory groups. Explicitly, a function $f:\ZZ^2\rightarrow \mathbb{C}$ in $\mathbb{C}\rtimes_\sigma\ZZ^2$ gets mapped to $\iota f$, where $(\iota f)(n)=f(n)1_\mathcal{A}, n\in\ZZ^2$. Also, an explicit representative for $[\zeta]\in K_1(C(S^1))$ is the (Fourier transform of) $\delta_1$, i.e.\ the function on the circle $\zeta:\varphi\mapsto e^{i\varphi}$; it is mapped under $\iota$ to the function $\iota\zeta:\varphi\mapsto e^{i\varphi}1_\mathcal{A}$ in $\mathcal{A}\rtimes_{\alpha_1}\ZZ^{(1)}$, and $\iota_*[\zeta]=[\iota\zeta]$.

Recall that $K_0(A_\theta) \cong \ZZ[\bf 1] \oplus \ZZ[\cP_\theta]$ where $[\bf 1]$ is the class of the rank 1 free module over $A_\theta$, and $[\cP_\theta]$ is the class of the Rieffel projection in $A_\theta$. For notational ease, we continue to denote the image of the class of the Rieffel projection under $\iota_*$ by $[\mathcal{P}_\theta]$, the image of the class of the rank-1 free module by $[\bf 1]$, and the image of the generator of $K^1(C(S^1))$ by $[\zeta]$.

\begin{proposition}\label{rieffeltozeta}
Under the boundary map of the Pimsner--Voiculescu exact sequence for the second crossed product by $\ZZ^{(2)}$, (the image of) $[\mathcal{P}_\theta]$ is mapped to (the image of) $[\zeta]$, while (the image of) $[\bf 1]$ is mapped to $[\bf 0]$.
\end{proposition}
\begin{proof}
This is true when $\mathcal{A}=\mathbb{C}$, where $\alpha$ is homotopic to the trivial action, so that the PV exact sequence \eqref{PVexactsequence} is
\beq
\xymatrix{ K_0(C(S^1)) \cong \ZZ[{\bf 1}] \ar[r]^0 & K_0(C(S^1)) \cong \ZZ[{\bf 1}] \ar[r]^{j_*} & K_0(A_\theta)\cong \ZZ[{\bf 1}]\oplus\ZZ[\mathcal{P}_\theta] \ar[d]^\partial & \\
  K_1(A_\theta)\cong \ZZ[\zeta]\oplus\ZZ \ar[u]^\partial & K_1(C(S^1))\cong\ZZ[\zeta] \ar[l]^{j_*} & K_1(C(S^1))\cong\ZZ[\zeta] \ar[l]^0 }\nonumber
\eeq
and $[\mathcal{P}_\theta]\xrightarrow{\partial}[\zeta]$, $[{\bf 1}]\xrightarrow{\partial}[{\bf 0}]$.

In general, the PV-exact sequence entails the sequence
\begin{equation*}
    K_0(\mathcal{A}\rtimes_{\alpha_1}\ZZ^{(1)})\xrightarrow{j_*}K_0(\mathcal{A}\rtimes_\sigma\ZZ^2)\xrightarrow{\partial}K_1(\mathcal{A}\rtimes_{\alpha_1}\ZZ^{(1)}),
\end{equation*}
exact in the middle. Here $j_*$ is the induced homomorphism under the inclusion $j$ of the first crossed product into the iterated crossed product. A preimage of $[\zeta]\in K_1(\mathcal{A}\rtimes_{\alpha_1}\ZZ^{(1)})$ under $\partial$ is $[\mathcal{P}_\theta]\in K_0(\mathcal{A}\rtimes_\sigma\ZZ^2),$ by considering the $\mathcal{A}=\mathbb{C}$ case and the definition of $\partial$ as an exponential map. The projective modules defined by the identity elements of the crossed products are mapped to each other under the unital map $j$, so $[{\bf 1}]\xrightarrow{j_*} [{\bf 1}]\xrightarrow{\partial}[\bf 0]$ by exactness.

\end{proof}

Note that $\alpha_2$ is homotopic to an automorphism whose restriction to the image of $C(S^1)$ in $\mathcal{A}\rtimes_{\alpha_1}\ZZ^{(1)}$ is trivial. Then the induced map $1-\alpha_{2*}$ on $K$-theory is zero on $[\zeta]\in K_1(\mathcal{A}\rtimes_{\alpha_1}\ZZ^{(1)})$. The PV exact sequence contains a short exact sequence
\begin{equation*}
    0\rightarrow K_0(\mathcal{A}\rtimes_{\alpha_1}\ZZ^{(1)})/\mathrm{Im}_{1-\alpha_{2*}}\xrightarrow{j_*}K_0(\mathcal{A}\rtimes_\sigma\ZZ^2)\xrightarrow{\partial}\mathrm{Ker}_{1-\alpha_{2*}}\left(K_1(\mathcal{A}\rtimes_{\alpha_1}\ZZ^{(1)})\right)\rightarrow 0,
\end{equation*}
verifying that $[\zeta]$ is in the image of $\partial$.

\begin{proposition}\label{rieffeltozetaspecial}
Let $\mathcal{A}=C(\Sigma)$ where $\Sigma$ is a Cantor set and suppose that the twisted $\ZZ^2$-action on $\Sigma$ by homeomorphisms is minimal. Then the boundary map of Proposition \ref{rieffeltozeta} maps $\ZZ[\mathcal{P}_\theta]$ isomorphically to $\ZZ[\zeta]$, and is zero on $K_0(C(\Sigma)\rtimes_\sigma\ZZ^2)/\ZZ[\mathcal{P}_\theta]$.
\end{proposition}
\begin{proof}
The PV exact sequence for the first $\ZZ^{(1)}$-action is
\beq
\xymatrix{ K_0(C(\Sigma))  \ar[r]^{1-\alpha_{1*}} & K_0(C(\Sigma)) \ar[r] & K_0(C(\Sigma)\rtimes_{\alpha_1}\ZZ^{(1)}) \ar[d]^{\partial_1} & \\
  K_1(C(\Sigma)\rtimes_{\alpha_1}\ZZ^{(1)}) \ar[u]^{\mathrm{Ind}} & K_1(C(\Sigma))\cong 0 \ar[l] & K_1(C(\Sigma))\cong 0 \ar[l] }.\nonumber
\eeq
In particular, $K_0(C(\Sigma)\rtimes_{\alpha_1}\ZZ^{(1)})$ is the group of $\alpha_1$-co-invariants of $K_0(C(\Sigma))\cong C(\Sigma,\ZZ)$ (see \eqref{cantorKgroups}), i.e. $K_0(C(\Sigma))/\mathrm{Im}_{1-\alpha_{1*}}(K_0(C(\Sigma)))$. Since
\beq
    0\longrightarrow K_1(C(\Sigma)\rtimes_{\alpha_1}\ZZ^{(1)})\xrightarrow{\;\;\mathrm{Ind}\;\;} \mathrm{Ker}_{1-\alpha_{1*}}\left(K_0(C(\Sigma))\right)\longrightarrow 0,\nonumber
\eeq
$K_1(C(\Sigma)\rtimes_{\alpha_1}\ZZ^{(1)})$ is isomorphic to the $\alpha_1$-invariant subgroup of $K_0(C(\Sigma))\cong C(\Sigma,\ZZ)$, and $K_1(C(\Sigma)\rtimes_{\alpha_1}\ZZ^{(1)})\supset \ZZ[\zeta]\xrightarrow{\;\;\mathrm{Ind}\;\;} \ZZ[{\bf 1}]$. The PV exact sequence for $\alpha_2$ 
includes the short exact sequence
\begin{equation}
\xymatrix{
    0 \ar[r] & K_0(C(\Sigma)\rtimes_{\alpha_1}\ZZ^{(1)})/\mathrm{Im}_{1-\alpha_{2*}} \ar[r] & K_0(C(\Sigma)\rtimes_\sigma\ZZ^2)
                \ar@{->} `r/8pt[d] `/10pt[l] `^dl[ll]|{\partial} `^r/3pt[dll] [dllr] \\
           & \mathrm{Ker}_{1-\alpha_{2*}} \left(K_1(C(\Sigma)\rtimes_{\alpha_1}\ZZ^{(1)})\right)  \ar[r]  & 0  }\label{sesforsecondaction}
\end{equation}
%
so $\partial$ (composed with ${\rm Ind}$) maps onto the $\alpha$-invariants of $C(\Sigma,\ZZ)$. Our assumption of minimality means that the only $\alpha$-invariant functions are the constant functions. From Proposition \ref{rieffeltozeta}, we already know that $[\mathcal{P}_\theta]\xrightarrow{\partial}[\zeta]$ for the boundary map for the second crossed product. Equation \eqref{sesforsecondaction} now simplifies to
\beq
    0\longrightarrow C(\Sigma,\ZZ)_{co}\longrightarrow C(\Sigma,\ZZ)_{co}\oplus\ZZ[\mathcal{P}_\theta]\xrightarrow{\;\;\partial\;\;} \ZZ[\zeta]\longrightarrow 0,\nonumber
\eeq
where we have written $C(\Sigma,\ZZ)_{co}$ for the (iterated) co-invariants of $C(\Sigma,\ZZ)$ under $\ZZ^2$.
\end{proof}

\section{Noncommutative T-duality}
\label{sect:FM}

T-duality describes an inverse mirror relationship between a pair of type II string theories. Mathematically, it is a geometric analogue of the 
Fourier transform, giving rise to a bijection of (Ramond--Ramond) fields and their charges, which belong to $K$-theory. This was carried 
out in the absence of a background flux in \cite{H}, in the presence of a background flux in \cite{BEM,BEM2}, and the noncommutative analog appears in \cite{MR05,MR06}. In this section, we give a refinement, showing the effect of noncommutative T-duality on generators of $K$-theory. This is relevant since, for example, the noncommutative torus $A_\theta$ appears as the bulk algebra when studying the IQHE, while a particular subalgebra of functions on the circle constitutes the boundary algebra, c.f.\ Section \ref{section:dualPV}.

Consider the diagram,
\begin{equation}
\xymatrix{ 
& {\cM_\theta} \ar[d] & \\
 &  C(\TT^2)\otimes A_\theta  &  \\
 C(\TT^2) \ar[ur]_{\iota_1} && A_\theta  \ar[ul]^{\iota_2}, 
 }\nonumber
\end{equation}
where 
$\iota_1, \iota_2$ are inclusion maps, and
$\cM_\theta$ is the universal finite projective module over $C(\TT^2)\otimes A_\theta$, playing the role of the Poincar\'e line bundle, which we now construct. Consider the central extensions
\begin{align*}
& 0\longrightarrow \mathrm{U}(1)\longrightarrow {\rm Heis}_\sigma^\RR \longrightarrow\RR^2\longrightarrow 0,\\
& 0\longrightarrow \mathrm{U}(1)\longrightarrow {\rm Heis}_\sigma^\ZZ \longrightarrow\ZZ^2\longrightarrow 0,
\end{align*}
determined by the multiplier $\sigma = \exp(2\pi i\theta\omega)$, where $\omega$ is the standard 
symplectic form on the vector space $\RR^2$ that restricts to $\ZZ^2$.
Then ${\rm Heis}_\sigma^\ZZ$ is a subgroup of ${\rm Heis}_\sigma^\RR$ and there is a left action 
of ${\rm Heis}_\sigma^\ZZ$ on $A_\theta = C^*(\ZZ^2, \sigma)$.
Set $\cM_\theta = C(\TT^2, \cV_\sigma)$, where 
\beq
\cV_\sigma =  {\rm Heis}_\sigma^\RR \times_{{\rm Heis}_\sigma^\ZZ} A_\theta \nonumber
\eeq 
is a locally trivial vector bundle with fibers $A_\theta$ over the quotient ${\rm Heis}_\sigma^\RR/{\rm Heis}_\sigma^\ZZ = \TT^2$.

Notice that $C(\TT^2)\otimes A_\theta $ and $A_\theta$ are $K$-oriented, so one has Poincar\'e duality in $K$-theory, 
\beq
PD_{C(\TT^2)\otimes A_\theta} : K_0(C(\TT^2)\otimes A_\theta) \cong K^0(C(\TT^2)\otimes  A_\theta),\quad
PD_{A_\theta} : K^0(A_\theta) \cong K_0(A_\theta). \nonumber
\eeq
Then noncommutative T-duality is the composition,
\beq
K_0(C(\TT^2))\ni [E] \longrightarrow 
\iota_2^!( (\iota_1)_*([E]) \otimes_{C(\TT^2)}  \cM_\theta) \in K_0(A_\theta), \nonumber
\eeq
where the wrong way map, or Gysin map $\iota_2^!\colon K_0(C(\TT^2)\otimes A_\theta) \to K_0(A_\theta) $
is defined by 
$\iota_2^! = PD_{A_\theta}\circ (\iota_2)^*\circ PD_{C(\TT^d)\otimes A_\theta}$, where 
$ (\iota_2)^*\colon K^0(C(\TT^2)\otimes A_\theta) \to K^0(A_\theta) $ is the homomorphism in $K$-homology.

\subsection{Poincar\'e duality in $K$-theory}

The 2D torus $\TT^2$ is a Spin manifold, therefore it is $K$-oriented. Poincar\'e duality in the $K$-theory of $\TT^2$ is given by 
\begin{align*}
PD_{\TT^2}: K^0(\TT^2) &\stackrel{\sim}{\longrightarrow} K_0(\TT^2)\\
[\mathcal{E}] & \longrightarrow [\dirac_{\TT^2}\otimes \mathcal{E}] 
\end{align*}
where  $\dirac_{\TT^2} \otimes \mathcal{E}$ is the Spin Dirac operator on $\TT^2$ coupled to the vector bundle $\mathcal{E}$.

In particular, we see that the class of the trivial line bundle $[\bf 1]$ maps to $ [\dirac_{\TT^2}]$, the class of the 
Spin Dirac operator on $\TT^2$, also known as the fundamental class in $K$-theory, which is a generator. Also the class of 
the prequantum line bundle $[\cL]$ maps to the class of the coupled Spin Dirac operator,  $ [\dirac_{\TT^2} \otimes \cL]$,
where we recall that the prequantum line bundle $\cL$ is associated to the principal circle bundle ${\rm Heis}_\sigma^\RR/{\rm Heis}_\sigma^\ZZ$ over $\TT^2$. That is, $\cL = {\rm Heis}_\sigma^\RR/{\rm Heis}_\sigma^\ZZ\times_{ \mathrm{U}(1)} \CC$. Equivalently, $\cL$ can be identified as the Poincar\'e line bundle over $\TT^2$, see section \ref{sub:realT-duality}.
Now 
\beq
K^0(\TT^2) \cong \ZZ[\bf 1] \oplus \ZZ[\cL]\nonumber
\eeq 
and since Poincar\'e duality is an isomorphism, we see that
\beq
K_0(\TT^2) \cong \ZZ[\dirac_{\TT^2}] \oplus \ZZ[\dirac_{\TT^2}\otimes \cL].\nonumber
\eeq 
Poincar\'e duality in the $K$-theory of $\TT^2$ exchanges $[\bf 1] \leftrightarrow [\dirac_{\TT^2}]$ 
and $[\cL] \leftrightarrow [\dirac_{\TT^2}\otimes \cL]$.\\

\subsection{Twisted Baum-Connes isomorphism}
Let $A_\theta$ be the noncommutative torus generated by a pair of unitaries $U_1, U_2$. Recall that the twisted Baum--Connes map \cite{CHMM,Mathai99} is an isomorphism of groups, because the 
Baum--Connes conjecture with coefficients is true for $\ZZ^2$ (cf. \cite{BCH,Connes94}),
\beq
\mu_\theta \colon K_0(\TT^2) \stackrel{\sim}{\longrightarrow}  K_0(A_\theta).\nonumber
\eeq
It is given by 
\beq
K_0(\TT^2) \ni[\dirac_{\TT^2} \otimes \mathcal{E}] \mapsto {\rm index}_{A_\theta}(\dirac_{\TT^2} \otimes \mathcal{E} \otimes \cV_\sigma)
\in K_0(A_\theta).\nonumber
\eeq
Recall also that 
\beq
K_0(A_\theta) \cong \ZZ[\bf 1] \oplus \ZZ[\cP_\theta]\nonumber
\eeq 
where $[\bf 1]$ is the class of the rank 1 free module over $A_\theta$, and $[\cP_\theta]$ is the class of the Rieffel projection in $A_\theta$.

Let $\delta_1, \delta_2$ 
be the (closed) derivations given by $\delta_i(U_j)=\delta_{ij}U_j$, where we drop various factors of $2\pi i$ for simplicity. Let the smooth noncommutative torus $A^\infty_\theta$ be defined as the intersection of the domains of arbitrary powers of the derivations. It is known that the periodic cyclic cohomology of the smooth noncommutative torus $A^\infty_\theta$ has generators (c.f.\ Thm.\ 53 of \cite{Connes85}, and \cite{Connes94})
\beq	
HP^{even}(A_\theta^\infty) = \CC[\tau] \oplus \CC[\tau_K],\nonumber
\eeq
where $\tau$ is the von Neumann trace on $A_\theta^\infty$ and $\tau_K$ is the area cocycle (also called the Kubo conductivity cocycle, or the noncommutative Chern character), 
\beq
\tau_K(a,b,c) = 
{\rm tr}(a(\delta_1(b)\delta_2(c) - \delta_2(b)\delta_1(c)))\label{areacocycle}
\eeq 
for $a, b, c \in A_\theta^\infty$.

\begin{proposition}\label{BC}
$\mu_\theta$ exchanges $[\dirac_{\TT^2}\otimes ([\cL] - [{\bf 1}])]  \leftrightarrow [\bf 1]$ and $ [\dirac_{\TT^2}]  \leftrightarrow [\cP_\theta]$.
\end{proposition}
\begin{proof}
By the index theorem in \cite{CHMM,Mathai99}, we see that 
\beq
\tau(\mu_\theta([\dirac_{\TT^2}\otimes \mathcal{E}])) = \int_{\TT^2} e^\B \wedge {\rm Ch}(\mathcal{E}) = {\rm rank}(\mathcal{E})  \int_{\TT^2} \B + c(\mathcal{E}) =  {\rm rank}(\mathcal{E})\,  \theta + c(\mathcal{E}) \nonumber
\eeq
where $\tau$ denotes the von Neumann trace on $A_\theta$ and $\B = \theta\, dx\wedge dy$ is a closed 2-form on $\TT^2$. Therefore 
\beq
\tau(\mu_\theta([\dirac_{\TT^2}\otimes ([\cL] - [{\bf 1}])])) =1 = \tau([{\bf 1}]),\qquad
\tau(\mu_\theta([\dirac_{\TT^2}\otimes {\bf 1}])) =\theta = \tau([\cP_\theta]).\label{traceidentity}
\eeq
If $\theta\not\in \QQ$ is irrational,  then the noncommutative torus $A_\theta$ is a factor (cf. \cite{Rieffel}) so that 
the von Neumann trace is injective, and therefore we conclude that
\beq
\mu_\theta([\dirac_{\TT^2}\otimes ([\cL] - [{\bf 1}])]) = [{\bf 1}], \qquad \mu_\theta([\dirac_{\TT^2}\otimes{\bf 1}]) =[\cP_\theta].\label{exchangeofgenerators}
\eeq
If $\theta$ is rational, the trace is no longer injective, but we can use the higher twisted index theorem \cite{Marcolli},
\beq
    \tau_K(\mu_\theta([\dirac_{\TT^2}\otimes \mathcal{E}]))=\int_{\TT^2}\omega\wedge e^{\bf B}\wedge {\rm Ch}(\mathcal{E})={\rm rank}(\mathcal{E}),\label{twistedindexformula}
\eeq
where $\omega$ is the standard invariant symplectic form on the vector space $\RR^2$ that descends to the torus $\TT^2$,
and where we have written $\tau_K(P)\equiv\tau_K(P,P,P)$ for the area cocycle defined in \eqref{areacocycle}. This leads to
\beq
    \tau_K(\mu_\theta([\dirac_{\TT^2}\otimes ([\cL] - [{\bf 1}])]))=0=\tau_K([{\bf 1}]),\qquad\tau_K(\mu_\theta([\dirac_{\TT^2}\otimes {\bf 1}]))=1=\tau_K([\mathcal{P}_\theta]),\label{highertraceidentity}
\eeq
which when considered alongside \eqref{traceidentity}, leads again to \eqref{exchangeofgenerators}.

\end{proof}

\subsection{Noncommutative T-duality}
Consider the composition, 
\beq
\mu_\theta \circ PD_{\TT^2} = T_{\TT^2}\colon K^0(\TT^2) \stackrel{\sim}{\longrightarrow} K_0(A_\theta)\label{noncommutativeTduality}
\eeq
which exchanges  $[\bf 1]  \leftrightarrow [\cP_\theta]$ and $([\cL] - [\bf 1])  \leftrightarrow [\bf 1]$ by Proposition \ref{BC}. This is noncommutative T-duality and is just the Connes--Thom isomorphism \cite{Connes81}, c.f.\ Section \ref{section:generalisations}.\\

\subsection{T-duality trivializes the bulk-boundary correspondence}\label{section:dualitytrivializes}

We will prove the commutativity of the diagram,
\beq
\xymatrix{
K^0(\TT^2)  \ar[d]^{\iota^*} \ar[r]^{\sim}_{T_{\TT^2}} & K_0(A_\theta) \ar[d]^\partial \\
K^0(\TT) \ar[r]^\sim_{T_\TT} &K^{-1}(\TT)  }   \label{commutativeTdualtorus}
\eeq
where the bottom horizontal arrow is the analogue of \eqref{noncommutativeTduality} for $\TT$.
That is, $T_\TT$ is determined by
\beq
K^0(\TT) \ni [{\bf 1}] \leftrightarrow [\zeta]\in K^{-1}(\TT),\nonumber
\eeq 
where $[\zeta] \in K^{-1}(\TT)$ is the generator. Recall that 
\beq
K^0(\TT) \cong \ZZ[{\bf 1}],\qquad K^{-1}(\TT) \cong \ZZ[\zeta].\nonumber
\eeq 
Let us check the commutativity of \eqref{commutativeTdualtorus} above.
We know that 
\beq
 \partial \circ T_{\TT^2}([{\bf 1}])=\partial \circ  \mu_\theta \circ PD_{\TT^2}([\bf 1]) = \partial  ([\cP_\theta]) = [\zeta]\nonumber
 \eeq
  by Proposition \ref{rieffeltozeta}.
Also, $$T_\TT \circ  \iota^*([{\bf 1}]) = T_\TT ([\bf 1]) = [\zeta].$$
Now  
\beq \partial \circ T_{\TT^2}([{\cL}]-[{\bf 1}])=\partial \circ  \mu_\theta \circ PD_{\TT^2}([{\cL}]-[{\bf 1}]) = \partial  ([\bf 1]) = [{\bf 0}].\nonumber
\eeq 
On the other hand,
\beq 
T_\TT \circ  \iota^*([{\cL}]-[{\bf 1}]) = T_\TT ([{\bf 0}]) = [{\bf 0}].\nonumber
\eeq
This verifies commutativity of the said diagram.\\

As a consequence, we see that upon applying T-duality, the relatively complicated boundary map,
\beq
\partial \colon K_0(A_\theta)  \to K^{-1}(\TT) \nonumber
\eeq
is equivalent to the much simpler restriction map,
\beq
\iota^* \colon K^0(\TT^2)  \to K^0(\TT). \nonumber
\eeq
The commutative version of this result (c.f.\ \eqref{commutativedualitydiagram}) is also true, and can be calculated using a Fourier--Mukai transform. This is explained in detail in Section \ref{section:bulkboundaryChern}.

\subsection{Cyclic cohomology, the Hall conductance and the winding number}\label{section:dualPV}

Recall that for the smooth noncommutative torus, we had
\beq	
HP^{even}(A_\theta^\infty) = \CC[\tau] \oplus \CC[\tau_K].\nonumber
\eeq
We also know that the periodic cyclic cohomology of the smooth functions on the circle $C^\infty(S^1)$ has generators,
\beq	
H^1(S^1) = HP^{odd}(C^\infty(S^1)) = \CC[\tau_w] \nonumber
\eeq
where $\tau_w$ is the winding number of a unitary. 

Using the Pimsner--Voiculescu exact sequence for cyclic cohomology \cite{Nest,Nistor}, we deduce exactly as in the $K$-theory case that
\beq
\partial: HP^{even}(A_\theta^\infty) \longrightarrow HP^{odd}(C^\infty(S^1)), \quad \partial(\tau_K)=\tau_w.\nonumber
\eeq
Finally, the results of \cite{Kellendonk1,Kellendonk2,Kellendonk3}. say that 
\beq
 \tau_K(P,P,P) = \partial \tau_K(\partial P, \partial P) = \tau_w(U, U),\nonumber
\eeq
where $P$ is a projection in $A_\theta^\infty$ and $U$ a unitary in $C^\infty(S^1)$ such that $\partial(P)=U$.

Now we point out how  to compute the trace and the conductance 2-cocycle. Let $P$ be a projection in $A_\theta\otimes\cK$. The class of $P$ is 
\beq
K_0(A_\theta) \ni [P] = m [{\bf 1}] + n [\cP_\theta], \nonumber
\eeq
for some integers $m, n\in\Z$. 
Then we compute the value of the von Neumann trace on $P$,
\beq
\tau(P)  = n\theta + m.\nonumber
\eeq
We also compute the value of the Kubo conductance cyclic 2-cocycle $\tau_K$ on $P$, 
\beq
\tau_K(P,P,P)  = n \tau_K(\cP_\theta, \cP_\theta, \cP_\theta) = n,\nonumber
\eeq
since $\tau_K$ is equal to $1$  on $\cP_\theta$ and $\tau_K$ is equal to zero on $\bf 1$.
This calculation can be done also on the torus $\TT^2$, using the higher twisted index theorem, 
section 2, \cite{Marcolli}, as follows. Consider $[\mathcal{E}] \in K^0(\TT^2)$ defined as
\beq
[\mathcal{E}] = n [{\bf 1}] + m [([\cL] - [{\bf 1}])].\nonumber
\eeq 
Then $T_{\TT^2}([\mathcal{E}])=P$ and we have the pairing
\beq
\int_{\TT^2}\omega\wedge e^\B \wedge {\rm Ch}(\mathcal{E}) = {\rm rank}(\mathcal{E})=n =   \tau_K(P,P,P).\nonumber
\eeq
On the other hand, if $U$ is a unitary such that $\partial(P)=U$, then $[U] = n[\zeta]$ and 
\beq
\tau_w(U,U) = n\tau_w(\zeta, \zeta) = n,\nonumber
\eeq
since the winding number of the unitary $\zeta$ is equal to 1. In particular, we see that 
\beq
\tau_K(P,P,P)  = \tau_w(U,U).\nonumber
\eeq

\subsection{Generalisations}\label{section:generalisations}

Let $\mathcal{A}$ be a unital $C^*$-algebra with an action of $\ZZ^2$, and $\sigma$ be a multiplier on $\ZZ^2$.
Then we want to compute the $K$-theory of the twisted crossed product $\mathcal{A}\rtimes_{\sigma} \ZZ^2$. 
By the Packer--Raeburn trick \cite{PR}, 
\beq
\mathcal{A}\rtimes_{\sigma} \ZZ^2 \otimes \cK \cong (\mathcal{A}\otimes \cK)\rtimes_{\alpha} \ZZ^2,\nonumber
\eeq
where $\alpha$ is an ordinary action related to $\sigma$.
 Recall that the induced algebra is defined as
 \beq
 {\rm Ind}_{\ZZ^2}^{\RR^2}(\mathcal{A}\otimes \cK, \alpha) = \{f:\RR^2 \to \mathcal{A}\otimes\cK| f(x+g) = \alpha(g)(f(x)),\, g\in\ZZ^2\},\nonumber
 \eeq
and has the property that it has an action of $\RR^2$ and there is a canonical Morita equivalence 
 \beq
  {\rm Ind}_{\ZZ^2}^{\RR^2}(\mathcal{A}\otimes \cK, \alpha) \rtimes \RR^2 \overset{\rm Morita}{\cong} (\mathcal{A}\otimes \cK)\rtimes_{\alpha} \ZZ^2.\nonumber
 \eeq
 Therefore by the Connes--Thom isomorphism theorem \cite{Connes81}, we conclude that
 \beq
 K_0(\mathcal{A}\rtimes_{\sigma} \ZZ^2) \cong K_0(  {\rm Ind}_{\ZZ^2}^{\RR^2}(\mathcal{A}\otimes \cK, \alpha)).\nonumber
 \eeq
 Notice that $ {\rm Ind}_{\ZZ^2}^{\RR^2}(\mathcal{A}\otimes \cK, \alpha) $ is just the space of sections of the 
flat  bundle $\cE=\RR^2\times_{\ZZ^2, \alpha} (\mathcal{A}\otimes \cK)$ over the torus $\TT^2$ with fiber $\mathcal{A}\otimes \cK$.
 
 
\begin{example} When $\mathcal{A}=\CC$, we have already analysed this case above, but let us say a few words. In this case, $\cE$ is just an algebra bundle 
 of compact operators over $\TT^2$, which are classified up to isomorphism by their Dixmier--Douady invariant in $H^3(\TT^2; \ZZ)=0$,
 so $\cE$ is trivializable, and therefore $K_0(  {\rm Ind}_{\ZZ^2}^{\RR^2}(\cK, \alpha)) \cong K^0(\TT^2)$ by the Morita invariance of $K$-theory.\\
 \end{example}
 
\begin{example} Let $\mathcal{A}=C(\Sigma)$ with $\Sigma$ a Cantor set, which is a compact Hausdorff space. Suppose $\Sigma$ is equipped with a mimimal action of $\ZZ^2$, and let $\sigma$
be a multiplier on $\ZZ^2$. The induced algebra  $ {\rm Ind}_{\ZZ^2}^{\RR^2}(\mathcal{A})$ in this case is just the double suspension $C(X)$, where 
 $X=\Sigma \times_{\ZZ^2} \RR^2$. Then we have the strong Morita equivalence,
 \beq
 C(\Sigma) \rtimes_{\sigma} \ZZ^2 \overset{\rm Morita}{\cong} C(X) \rtimes_{\sigma} \RR^2\nonumber
 \eeq
which follows from the Packer--Raeburn trick \cite{PR} and from  \cite{Green,Rieffel82},
  \beq
 C(\Sigma) \otimes \cK \rtimes_{\alpha} \ZZ^2 \overset{\rm Morita}{\cong} C(X) \otimes \cK  \rtimes_{\alpha} \RR^2.\nonumber
 \eeq
Therefore by the Connes--Thom isomorphism theorem \cite{Connes81} and the Morita invariance of $K$-theory, we see that
\beq
K_\bullet(C(\Sigma) \rtimes_{\sigma} \ZZ^2) \cong K_\bullet(C(X) \rtimes_{\sigma} \RR^2) \cong K_\bullet(C(X)).\label{mappingtorus}
\eeq
In particular, we see that 
\beq
K_\bullet(C(\Sigma) \rtimes_{\sigma} \ZZ^2) \cong K_\bullet(C(\Sigma) \rtimes \ZZ^2).\nonumber
\eeq

It is known (cf.\ \cite{Putnam}) that 
\beq
K_0(C(\Sigma)) = C(\Sigma, \Z), \qquad K_1(C(\Sigma)) =\{0\}.\label{cantorKgroups}
\eeq
Now the authors of \cite{BCL} have computed that
\beq
K_0(C(\Sigma) \rtimes \ZZ^2)\cong C(\Sigma, \Z)_{co} \oplus \Z,\nonumber
\eeq
where $C(\Sigma, \Z)_{co} $ are the co-invariants under the $\Z^2$-action, while we showed in Proposition \ref{rieffeltozetaspecial} that
\beq
K_0(C(\Sigma) \rtimes_\sigma \ZZ^2)\cong C(\Sigma, \Z)_{co} \oplus \Z[\cP_\theta].\nonumber
\eeq
More precisely:
\begin{itemize}
        \item    The natural
inclusion
                $A_\theta =C^{*}_{r}(\Z^2, \sigma)\hookrightarrow C(\Sigma)\rtimes_\sigma\Z^2\,$, takes the 
Rieffel projection $\cP_\theta$ to a projection in $C(\Sigma)\rtimes_\sigma\Z^2\,$ which generates the $\Z$
factor in $K_0(C(\Sigma)\rtimes_\sigma\Z^2)$.
                    \item   The inclusion $C(\Sigma,\Z)_{co} \hookrightarrow K_0(C(\Sigma)\rtimes_\sigma\Z^2)$
is induced by
                the inclusion $C(\Sigma)\hookrightarrow C(\Sigma)\rtimes_\sigma\Z^2$.
\end{itemize}

Recall from Proposition \ref{rieffeltozetaspecial} and its proof that
\beq
 K_1(C(\Sigma) \rtimes \ZZ)\cong C(\Sigma,\ZZ)^{\ZZ} \supset C(\Sigma,\ZZ)^{\ZZ^2}=\ZZ[{\bf 1}_\Sigma]={\rm Ind}(\ZZ[\zeta]),\nonumber
\eeq
where $C(\Sigma,\ZZ)^{\ZZ}$ is the subgroup of $C(\Sigma,\ZZ)\cong K_0(C(\Sigma))$ invariant under a single $\ZZ$ factor, ${\bf 1}_\Sigma$ is the constant function $\Sigma\mapsto 1$, and 
\beq
\partial([\cP_\theta]) = [\zeta], \qquad \partial(C(\Sigma, \Z)_{co}  )=[{\bf 0}].\label{actionofboundarymap}
\eeq

Armed with these calculations, we can verify as before the commutativity of the diagram,
\beq
\xymatrix{
K^0(X)  \ar[d]^{\iota^*} \ar[r]^{\sim\qquad}_{T\qquad} & K_0(C(\Sigma)\rtimes_{\sigma} \ZZ^2) \ar[d]^\partial \\
K^0(X_1) \ar[r]^{\sim\qquad}_{T_1\qquad} & K_1(C(\Sigma)\rtimes  \ZZ)  }   \label{cantorcommutativediagram}
\eeq
where $X_1=\Sigma\times_\ZZ\RR$ is the suspension which includes into the double suspension $X_1\stackrel{\iota}{\hookrightarrow}X$, and $T,T_1$ are the T-duality isomorphisms.

First, we note that a vector bundle $\mathcal{E}$ over $X$ has constant rank everywhere, since $X$ is connected by Lemma 3, \cite{BO07}. 
Then the rank gives a natural splitting $K^0(X)=\widetilde{K^0}(X)\oplus\ZZ[{\bf 1}_X]$ where ${\bf 1}_X$ denotes the trivial line bundle over $X$. Also, there is a T-duality isomorphism $K^0(X_1)\cong K_1(C(\Sigma)\rtimes\ZZ)\cong C(\Sigma,\ZZ)^\ZZ$ by the same arguments leading to Eq.\ \eqref{mappingtorus} but applied to a crossed product with a single $\ZZ$ rather than $\ZZ^2$. Elements of $C(\Sigma,\ZZ)^\ZZ$ are integer linear combinations of characteristic functions on $\ZZ$-invariant clopen subsets of $\Sigma$. In particular, ${\bf 1}_\Sigma={\rm Ind}([\zeta])$ corresponds to (the $K$-theory class of) the trivial line bundle ${\bf 1}_{X_1}\rightarrow X_1$ which generates a subgroup $\ZZ[{\bf 1}_{X_1}]$ in $K^0(X_1)$. Since a non-zero $\mathcal{E}$ is supported on all of $\Sigma$, the restriction homomorphism $\iota^*$ lands on this subgroup, taking $[\mathcal{E}]$ to ${\rm rank}(\mathcal{E})[{\bf 1}_{X_1}]$. In summary,
\beq
    T_1\circ \iota^*([\mathcal{E}])={\rm rank}(\mathcal{E})[\zeta].\label{tdualofrestriction}
\eeq

Next, we use a normalized $\ZZ^2$-invariant measure $\mu$ on $\Sigma$ to construct a cyclic 2-cocycle $\tau_{K,\mu}$ on $C(\Sigma)\rtimes_\sigma \ZZ^2$, given by
\beq
    \tau_{K,\mu}(a,b,c)=\int_\Sigma {\rm tr}\Big(a\big(\delta_1(b)\delta_2(c)-\delta_2(b)\delta_1(c)\big)(\vartheta)\Big)\,d\mu(\vartheta)\nonumber
\eeq
for $a,b,c$ in a dense subalgebra of $C(\Sigma)\rtimes_\sigma \ZZ^2$, mimicking \eqref{areacocycle}. It is clear that
\beq
\tau_{K,\mu}([\mathcal{P}_\theta])\equiv\tau_{K,\mu}(\mathcal{P}_\theta,\mathcal{P}_\theta,\mathcal{P}_\theta)=\tau_K(\mathcal{P}_\theta,\mathcal{P}_\theta,\mathcal{P}_\theta)=1.\nonumber
\eeq
Furthermore, $\tau_{K,\mu}$ vanishes on $C(\Sigma,\ZZ)_{co}\cong K_0(C(\Sigma)\rtimes\ZZ)_{co}\subset K_0(C(\Sigma)\rtimes_\sigma\ZZ^2)$, since the representative elements of $K_0(C(\Sigma)\rtimes\ZZ)$ are annihilated by $\delta_2$. An index calculation along the lines of \eqref{twistedindexformula}, but using the higher twisted foliated index theorem instead (see Appendix \ref{appendix:indextheorem}) gives $\tau_{K,\mu}(T([\mathcal{E}]))={\rm rank}(\mathcal{E})$, from which we deduce that
\beq
    T([\mathcal{E}])={\rm rank}(\mathcal{E})[\mathcal{P}_\theta]+\mathcal{C},\qquad \mathcal{C}\in C(\Sigma,\ZZ)_{co}.\nonumber
\eeq
It follows from \eqref{actionofboundarymap} that
\beq
    \partial\circ T([\mathcal{E}])=\partial\big({\rm rank}(\mathcal{E})[\mathcal{P}_\theta]+\mathcal{C}\big)={\rm rank}(\mathcal{E})[\zeta],\nonumber
\eeq
which together with \eqref{tdualofrestriction}, shows that \eqref{cantorcommutativediagram} commutes.

\end{example}

\subsection{Real T-duality in the commutative case}\label{sub:realT-duality}
As we will use T-duality for the real $K$-theory of tori in the subsequent section on time-reversal invariant insulators, we give an outline of its construction, which is explained in more detail in \cite{H}. The geometric analog of the Fourier transform is the Fourier--Mukai transform, which may be summarised by the diagram
\begin{equation}
\xymatrix{ 
& {\mathcal P} \ar[d] & \\
 &  \TT^d\times \widehat{\TT}^d \ar[dl]_{p} \ar[dr]^{\widehat{p}} &  \\
\TT^d && \widehat{\TT}^d. 
 }\nonumber
\end{equation}
Here the T-dual $d$-torus to $\TT^d=\RR^d/\ZZ^d$ is denoted by $\widehat{\TT}^d$ for emphasis, and we identify $\widehat{\TT}^d$ with the character space of $\ZZ^d$, i.e.\ $\widehat{\TT}^d\ni k:n\mapsto k(n)\coloneqq e^{in\cdot k}$ for $n\in\ZZ^d$, while
$\mathcal{P}$ is the Poincar\'{e} line bundle over $\TT^d\times\widehat{\TT}^d$. In the complex case, $\mathcal{P}$ can be realized as the quotient of $\mathbb{R}^d\times\widehat{\TT}^d\times\mathbb{C}$ under the action of $n\in\ZZ^d$ given by $n\cdot(x,k,z)\mapsto (x+n,k,k(n)z)$. In the real case, $\TT^d$ has the trivial involution whereas $\widehat{\TT}^d$ is given the involution $k\mapsto -k$ corresponding to complex conjugation of characters. Then $\mathcal{P}$ has a lift of this involution to an antilinear involution $\Theta:(x,k,z)\mapsto(x,-k,\overline{z})$, since $\Theta(n\cdot(x,k,z))=(x+n,-k,\overline{k(n)z})=(x+n,-k,(-k)(n)\overline{z})=n\cdot(\Theta(x,k,z))$. The Fourier--Mukai transform is $\widehat{p}_*(p^*(\mathcal E)\otimes \mathcal P)$ for $\mathcal{E}$ representing a class in $KO^{-n}(\TT^d)$, giving the isomorphisms
\beq 
KO^{-n}(\TT^d)\cong  KR^{-n+d}(\widehat{\TT}^d).\nonumber
\eeq
More general dualities exist amongst other variants of real $K$-theory. They are explained in the context of the real Baum--Connes conjecture, in analogy to \eqref{noncommutativeTduality}, in \cite{Rosenberg,Baum}.

\section{T-duality and bulk-boundary for insulators}
\subsection{Bulk-boundary for Chern insulators}\label{section:bulkboundaryChern}
The sections of the valence bundle $\mathcal{E}$ in a 2D Chern insulator form a finitely-generated projective (f.g.p.) module for the algebra of functions $C(\TT^2)$, and determines a class in $K_0(C(\TT^2))$. We interpret $C(\TT^2)$ as a trivial crossed product (or group $C^*$-algebra) $C(S^1)\rtimes_\mathrm{id}\ZZ^{(2)}$, where $C(S^1)\cong\mathbb{C}\rtimes_\mathrm{id}\ZZ^{(1)}$ is itself the group $C^*$-algebra for $\ZZ^{(1)}$ interpreted as longitudinal translations along a boundary, and $\ZZ^{(2)}$ is a group of transverse translations. 
The Toeplitz extension \eqref{toeplitzextension} associated to the (trivial) crossed product by $\ZZ^{(2)}$ is
{\small{\beq
\xymatrix{0\ar[r] & \left(C(S^1)\otimes C_0(\ZZ)\right)\rtimes_{\mathrm{id}\otimes\tau}\ZZ^{(2)}\ar[r] & \left(C(S^1)\otimes C_0(\ZZ\cup\{\infty\})\right)\rtimes_{\mathrm{id}\otimes\tau}\ZZ^{(2)}
     \ar@{->} `r/8pt[d] `/10pt[l] `^dl[ll]|{} `^r/3pt[dll] [dllr] \\
    & C(\TT^2)\ar[r] & 0  \nonumber }
\eeq}}
where $\tau$ is left translation on $\ZZ$ and fixes $\infty$. It can more consisely be written as
\beq
    0\longrightarrow C(S^1)\otimes\mathcal{K}\longrightarrow \mathcal{T}(C(S^1),\mathrm{id})\longrightarrow C(\TT^2)\longrightarrow 0.\nonumber
\eeq

The PV exact sequence is
{\small{\beq
\xymatrix{ K_0(C(S^1)) \cong \ZZ \ar[r]^0 & K_0(C(S^1)) \cong \ZZ[{\bf 1}] \ar[r] & K_0(C(\TT^2))\cong \ZZ[{\bf 1}]\oplus\ZZ([\cL] - [{\bf 1}]) \ar[d]^\partial & \\
  K_1(C(\TT^2)\cong \ZZ\oplus\ZZ \ar[u] & K_1(C(S^1))\cong\ZZ[\zeta] \ar[l] & K_1(C(S^1))\cong\ZZ[\zeta] \ar[l]^0 \nonumber}
\eeq}}
whence we see that the boundary map $\partial$ maps $[\mathcal{L}]-[{\bf 1}]$ with Chern number 1 to $[\zeta]$ with winding number 1 in $K$-theory, whilst killing the trivial bundles $\ZZ[\bf 1]$. 

The (commutative) T-dual of $\mathcal{E}$ can be described more explicitly. Let $\widehat{\TT}^2$ be a ``dual'' torus with coordinates $k_1',k_2'$. The \emph{Poincar\'{e} line bundle} $\mathcal{P}$ over $\TT^2\times\widehat{\TT}^2$ has first Chern class $c_1(\mathcal{P})=\mathrm{d}k_1\wedge\mathrm{d}k_2'+\mathrm{d}k_2\wedge\mathrm{d}k_1'$. The T-dual to (or the Fourier--Mukai transform of) $\mathcal{E}$ is the bundle $\widehat{\mathcal{E}}$ over $\widehat{\TT}^2$ whose Chern character is \cite{H}
\begin{equation}
\mathrm{Ch}(\widehat{\mathcal{E}})=\int_{\TT^2}\mathrm{Ch}(\mathcal{E})\mathrm{Ch}(\mathcal{P})=c(\mathcal{E})+\mathrm{rank}(\mathcal{E})\mathrm{d}k_1'\wedge\mathrm{d}k_2',\nonumber
\end{equation}
showing that the rank and the Chern number are interchanged. This is the commutative version of \eqref{noncommutativeTduality}.

The bulk-boundary homomorphism is again trivialized after applying T-duality:
\beq
\xymatrix{
K^0(\TT^2) \cong \ZZ[{\bf 1}]\oplus\ZZ([\cL] - [{\bf 1}]) \ar[d]^{\iota^*} \ar[r]^{\sim} & K^0(\TT^2) \cong \ZZ([\cL] - [{\bf 1}]) \oplus \ZZ[{\bf 1}]\ar[d]^\partial \\
K^0(\TT) \cong \ZZ[{\bf 1}]\ar[r]^{\quad\sim} & K^{-1}(\TT) \cong \ZZ[\zeta] }.\label{commutativedualitydiagram}
\eeq
This is the simplest example illustrating the principle that $\partial\leftrightarrow i^*$ under T-duality: in the crossed product, the algebra is trivial and the $\ZZ^2$-action is trivial and untwisted.

Physically, the Chern number of $\mathcal{E}$ relates to a quantized transverse conductance through the Thouless--Kohmoto--Nightingale--den Nijs (TKNN) formula, while the winding number under the boundary map counts the number of gapless chiral edge modes.

In general, if we let $\TT^d=\RR^d/\ZZ^d$ be a fundamental domain for the real space translations by $\ZZ^d$, then the T-dual torus is the bulk Brillouin torus of quasi-momenta for $\ZZ^d$. Similarly, the translations $\ZZ^{d-1}$ along a boundary have a fundamental domain $\TT^{d-1}$ whose T-dual torus is the corresponding ``boundary Brillouin torus''. The bulk-boundary homomorphism $\partial$ maps a bulk momentum space invariant to a boundary momentum space invariant. On the other hand, $i^*$ is simply a geometric restriction map between $K$-theory groups of the bulk and boundary fundamental domains. This interpretation of $\partial\leftrightarrow i^*$ under T-duality is consistent with, and generalizes the physical intuition of the bulk-boundary correspondence from a semiclassical picture of the integer quantum Hall effect: the electron's cyclotron orbits are geometrically restricted/intercepted by a boundary, so that closed circular orbits in the bulk become chiral edge channels along the boundary. This discussion is captured schematically by the following diagram:

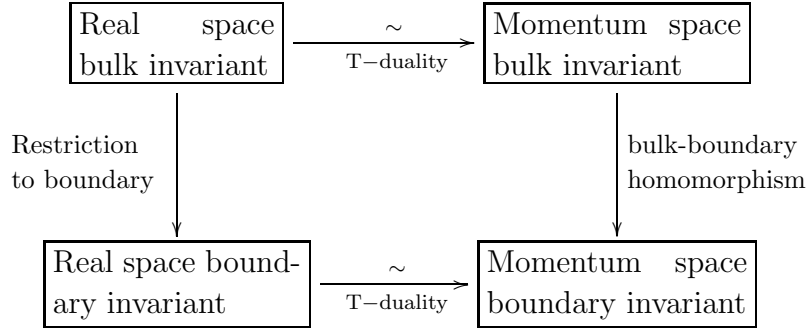
\beq\label{metadiagram}\nonumber
\xymatrix{
\framebox{\parbox{6.3em}{Real space bulk invariant}}  \ar[dd]_{\parbox{5em}{\footnotesize Restriction to boundary}} \ar[rr]^{\sim\;}_{\rm T-duality} && \framebox{\parbox{8em}{Momentum space bulk invariant}} \ar[dd]^{\parbox{7.7em}{\footnotesize bulk-boundary homomorphism}} \\ && \\
\framebox{\parbox{8em}{Real space boundary invariant}}  \ar[rr]^{\sim}_{\rm T-duality} && \framebox{\parbox{8.4em}{Momentum space boundary invariant}}
}
\eeq

\subsection{Time reversal, Quaternionic bundles, and real crossed products}\label{section:quaternionicbundles}
The standard vector bundle model for 2D band insulators with fermionic time-reversal symmetry $T$ (implemented by an antilinear operator $\mathsf{T}$ with $\mathsf{T}^2=-1$) is based on Quaternionic bundles $\mathcal{E}$ over $\widehat{\TT}^2$ (see \cite{deNittis,Luke} for a detailed discussion of the bundle theory). Here $\widehat{\TT}^2$ has the orientation-preserving involution $\varsigma: k\mapsto -k$ with four fixed points. Quaternionic bundles are complex vector bundles equipped with a Quaternionic structure $\mathsf{T}, \mathsf{T}^2=-1$ mapping the fiber over $k$ antilinearly to that over $-k$. At the fibers over the fixed points of the involution on $\widehat{\TT}^2$, $\mathsf{T}$ reduces to a genuine quaternionic structure, and so $\mathcal{E}$ has even rank as a complex vector bundle.

It is convenient to take an equivalent and perhaps more fundamental view, which is that the sections of a Quaternionic $\mathcal{E}$ form a f.g.p.\ module for the \emph{real} $C^*$-algebra $(\mathbb{R}\rtimes_\mathrm{id}\ZZ^2)\otimes_\mathbb{R}\mathbb{H}$, which can be understood as the $C^*$-algebra generated by the symmetry group $\ZZ^2\times\{1,T\}$. Here, $T$ denotes time-reversal, and is required to be implemented antilinearly and as an anti-involution. In more detail, we want to augment the complex group $C^*$-algebra of $\ZZ^2$, namely $\mathbb{C}\rtimes_\mathrm{id}\ZZ^2\cong C(\TT^2)$, to include $\mathsf{T}$. So we first rewrite $\mathbb{C}\rtimes_\mathrm{id}\ZZ^2$ as $(\mathbb{R}\rtimes_\mathrm{id}\ZZ^2)\otimes_\mathbb{R}\mathbb{C}$, where 
\beq
\mathbb{R}\rtimes_\mathrm{id}\ZZ^2\coloneqq C(i\TT^2)\equiv\{f:\TT^2\rightarrow\mathbb{C}\,:\,\overline{f(k)}=f(-k)\,\,\forall k\in\TT^2\}.\nonumber
\eeq
Then $\mathsf{T}$ introduces a quaternionic structure, augmenting the $\mathbb{C}$-factor to a $\mathbb{H}$-factor.

Thus, we are interested in the real operator $K$-theory group 
\beq
KO_0(C(i\TT^2)\otimes_\mathbb{R}\mathbb{H})\cong KO_4(C(i\TT^2))\cong KO_4(\mathbb{R})\oplus KO_2(\mathbb{R})=\ZZ\oplus\ZZ_2,\nonumber
\eeq
generated by f.g.p.\ modules for $C(i\TT^2)\otimes_\mathbb{R}\mathbb{H}$, for which the sections of a Quaternionic $\mathcal{E}$ furnish an example. 
We can also compute in topological $K$-theory, e.g.\ using \eqref{binomialformula},
\beq
KO_0(C(i\TT^2)\otimes_\mathbb{R}\mathbb{H})\cong KQ^0(\widehat{\TT}^2)\cong KQ^0(\star)\oplus\widetilde{KQ}^0(\widehat{\TT}^2)\cong \ZZ\oplus\ZZ_2.\nonumber
\eeq
Here $\star$ denotes a fixed point for the involution $\varsigma$ on the torus $\widehat{\TT}^2$, and the reduced $\widetilde{KQ}^0(\widehat{\TT}^2)$ is the kernel of the induced map under the inclusion of $\star$ into $\widehat{\TT}^2$. Thus, the $\ZZ$ counts the Quaternionic rank of $\mathcal{E}$, while $\ZZ_2$ is related to the invariant introduced in the physics literature by Fu--Kane--Mele \cite{KM2,FK2,FKM}, which we will identify as the Chern number of a complex sub-bundle $\mathcal{E}$ modulo 2.

We note that a left $\mathbb{H}$-module is a module for $\mathbb{C}\oplus\mathbb{C}j\cong \mathbb{C}^2$ (after choosing a copy of $\mathbb{R}\oplus i\mathbb{R}\cong \mathbb{C}\subset\mathbb{H}$) along with multiplication by $j$ effecting $(z,w)\mapsto(-\overline{w},\overline{z})$ (a quaternionic structure). Also, the map $V:(z,w)\mapsto(-w,z)$ is complex-linear and commutes with the $\mathbb{H}$-action. A f.g.p.\ $C(i\TT^2)\otimes_\mathbb{R}\mathbb{H}$-module is in particular a f.g.p.\ $C(i\TT^2)\otimes_\mathbb{R}(\mathbb{C}\oplus\mathbb{C}j)\cong C(\widehat{\TT}^2)\oplus C(\widehat{\TT}^2)$-module, and is thus a direct sum of sections of two \emph{complex} bundles. Furthermore, the map $\mathsf{T}=1\otimes j$ provides an invertible anti-linear map between the two complex sub-bundles which preserves the orientation on the base space, whence they have the same rank but opposite Chern numbers. Alternatively, Proposition 4.3 of \cite{deNittis} says that Quaternionic bundles over $\widehat{\TT}^2$ are trivial when viewed simply as complex bundles. Physically, we can think of the two (non-canonically defined) complex sub-bundles as Kramers partners, which are exchanged under $\mathsf{T}$. Note that neither the complex isomorphism classes of the two complex sub-bundles nor the order in which they are presented as direct summands need to be invariants of the total bundle as a Quaternionic bundle.

\begin{example}
The basic trivial Quaternionic bundle ${\bf 1}_Q$ has the free module of sections $C(i\TT^2)\otimes_\mathbb{R}\mathbb{H}$. It is the complex bundle $\widehat{\TT}^2\times\mathbb{C}^2=(\widehat{\TT}^2\times\mathbb{C})\oplus(\widehat{\TT}^2\times\mathbb{C})$ on which the Quaternionic structure $\mathsf{T}$ is $\mathsf{T}(k,(z,w))=(-k,(-\overline{w},\overline{z}))$. The analogous trivial Quaternionic bundle over $\widehat{\TT}^d$ for any $d\in\mathbb{N}$ will also be denoted by ${\bf 1}_Q$.
\end{example}

\begin{example}
There is a basic non-trivial Quaternionic bundle $\widehat{\mathcal{L}}$ over $\widehat{\TT}^2$ which is not isomorphic (in the category of Quaternionic bundles over $\widehat{\TT}^2$) to the basic trivial one. As a complex bundle, it is $\widehat{\mathcal{L}}=\varsigma^*\mathcal{L}\oplus\overline{\mathcal{L}}$, with fiber over $k$ being $\mathcal{L}_{-k}\oplus\overline{\mathcal{L}_k}$, and Quaternionic structure given fiberwise by $\mathsf{T}|_{\widehat{\mathcal{L}}_k}:(z_1,\overline{z_2})\mapsto(-z_2,\overline{z_1})$. That $\widehat{\mathcal{L}}$ is indeed non-trivial is a special case of a general result for low-dimensional spaces with involution shown in Section 4 of \cite{deNittis}. The same construction generates a Quaternionic bundle of complex rank 2 from any complex line bundle $\mathcal{L}$ which may have any Chern number.
\end{example}

Note that we can apply $V$ fiberwise on $\widehat{\mathcal{L}}$ to obtain a Quaternionic bundle $V\widehat{\mathcal{L}}$ (with Quaternionic structure $V\mathsf{T}V^{-1}$). Two copies of $\widehat{\mathcal{L}}$ are thus isomorphic to $\widehat{\mathcal{L}}\oplus V\widehat{\mathcal{L}}=(\varsigma^*\mathcal{L}\oplus\overline{\mathcal{L}})\oplus (\overline{\mathcal{L}}\oplus\varsigma^*\mathcal{L})$ with Quaternionic structure interchanging the two bracketed terms. Since each of the bracketed terms has Chern number 0 and can be written as the sum of two trivial line bundles, we can decompose $\widehat{\mathcal{L}}\oplus V\widehat{\mathcal{L}}$ as two copies of ${\bf 1}_Q$. This suggests that for a general Quaternionic bundle, we can take the mod 2 Chern number of either of the two complex subbundles interchanged by the Quaternionic structure as a natural invariant. In fact, it is shown in \cite{deNittis} that the mod 2 Chern number of $\mathcal{L}$ in $\widehat{\mathcal{L}}$ characterizes the Quaternionic isomorphism class of $\widehat{\mathcal{L}}$, and is a special case of the ``FKMM invariant'' defined more generally for spaces with involution. For $\widehat{\TT}^2$, the FKMM invariant is $\ZZ_2$-valued, and completely characterizes the stable classes of Quaternionic bundles (complex rank 2 is already in the stable regime), thus coinciding with $\widetilde{KQ}^0(\widehat{\TT}^2)\cong\ZZ_2$. This invariant can also be computed through a fixed point formula, generalizing the formula introduced by Fu--Kane--Mele (FKM) \cite{KM2,FK2,FKM} (see also \cite{FM}). We will offer a different perspective on the FKMM/FKM $\ZZ_2$-invariant in 2D, by relating it to the second Stiefel--Whitney class for real bundles using T-duality.

\subsection{Bulk-boundary for time-reversal invariant insulators}
Recall that the sections of the valence bundle for a 2D time-reversal invariant insulators form a f.g.p.\ module for the $C^*$-algebra $C(i\TT^2)\otimes_\mathbb{R}\mathbb{H}$, which can be expressed as a (trivial real) crossed product with $\ZZ^{(2)}$ in the transverse direction to a boundary, as $(C(i\TT)\otimes_\mathbb{R}\mathbb{H})\rtimes_{\mathrm{id}}\ZZ^{(2)}\cong (C(i\TT)\rtimes_{\mathrm{id}}\ZZ^{(2)})\otimes_\mathbb{R}\mathbb{H}$. There are isomorphisms $KO_j(C(i\TT^2)\otimes_\mathbb{R}\mathbb{H})\cong KQ^{-j}(\widehat{\TT}^2)$ and $KO_j(C(i\TT)\otimes_\mathbb{R}\mathbb{H})\cong KQ^{-j}(\widehat{\TT})$. We can form the associated real Toeplitz extension and PV-exact sequence as in \eqref{toeplitzextension} and \eqref{PVexactsequence}. The boundary homomorphism which we are interested in is that from $KO_0(\cdot)$ to $KO_7(\cdot)$, which in terms of the topological $KQ$-theory is

\beq
\xymatrix{ \ldots KQ^{0}(\widehat{\TT}) \ar[r]^0 & KQ^{0}(\widehat{\TT})\cong \ZZ[{\bf 1}_Q] \ar[r] & KQ^{0}(\widehat{\TT}^2)) \cong \ZZ[{\bf 1}_Q]\oplus \ZZ_2[\widehat{\mathcal{L}}]\ar[d]^\partial & \\
 \ldots  KQ^{-7}(\widehat{\TT}^2)  & KQ^{-7}(\widehat{\TT}) \ar[l] & KQ^{-7}(\widehat{\TT}) \cong \ZZ_2 \ar[l]^0 }\nonumber
\eeq
making it clear that $\ZZ_2\xrightarrow{\partial}\ZZ_2$ isomorphically. The calculations for the $KQ$-theory groups of the torus with involution follow from the split short exact sequence
\beq
    0\longrightarrow KQ^{-j}(X,Y)\longrightarrow\widetilde{KQ}^{-j}(X)\longrightarrow\widetilde{KQ}^{-j}(Y)\longrightarrow 0\label{retractsequence}
\eeq
for $Y$ an equivariant retract of $X$, which also holds for real $K$-theory of spaces with trivial involution. Applying \eqref{retractsequence} this to the successive retracts $X\times Y\rightarrow (X\times Y)/X\rightarrow X\wedge Y$, we have
\beq
    \widetilde{KQ}^{-j}(X\times Y)\cong \widetilde{KQ}^{-j}(X\wedge Y)\oplus\widetilde{KQ}^{-j}(X)\oplus\widetilde{KQ}^{-j}(Y),\label{retractformula}
\eeq
which when applied to $X=\widehat{S^1}=Y$ gives\footnote{We use the notation $S^d$ for the $d$-sphere with trivial involution, and $\widehat{S^d}$ for the $d$-fold smash product of $\widehat{S^1}=\widehat{\TT}$ with itself, with the involution inherited from the $\widehat{S^1}$ factors. For example, $\widehat{S^2}$ is the 2-sphere in $\mathbb{R}^{1,2}$ with involution $(x,y,z)\mapsto(x,-y,-z)$.}
\beq
    \widetilde{KQ}^{-j}(\widehat{\TT}^2)\cong \widetilde{KQ}^{-j}(\widehat{S^2})\oplus\widetilde{KQ}^{-j}(\widehat{S^1})\oplus\widetilde{KQ}^{-j}(\widehat{S^1}).\nonumber
\eeq
There are Bott periodicity isomorphisms $KQ^{-j}(\cdot)
\cong KR^{-j\pm 4}(\cdot)$ and $KR^{p,q}(\cdot)\cong KR^{p+1,q+1}(\cdot)$ \cite{Dupont,Luke}, so each $\widetilde{KQ}$-group of $\widehat{S^d}$ is isomorphic to some real $K$-theory group of a point. For example, $\widetilde{KQ}^0(\widehat{S^2})\cong \widetilde{KR}^{-4}(\widehat{S^2})\cong \widetilde{KR}^0(S^2)\cong KO^{-2}(\star)= \ZZ_2$ while $\widetilde{KQ}^0(\widehat{\TT})\cong \widetilde{KR}^{-4}(\widehat{\TT})\cong \widetilde{KR}^0(S^3)\cong KO^{-3}(\star)= 0$. Thus $\widetilde{KQ}^{-0}(\widehat{\TT}^2)\cong \widetilde{KQ}^{-0}(\widehat{S^2})\cong\ZZ_2$. More generally, there are binomial formulae
\beq
    KQ^{-n\pm 4}(\widehat{\TT}^d)\cong KR^{-n}(\widehat{\TT}^d) \cong \bigoplus_{j=0}^d \left[KO^{-n+ j}(\star)\right]^{d\choose j},\label{binomialformula}
\eeq
which comes from iterating $KR^{-n}(\widehat{\TT}^d)\cong KR^{-n}(\widehat{\TT}^{d-1})\oplus KR^{-(n-1)}(\widehat{\TT}^{d-1})$ for example.

T-duality maps the $KQ$-groups (using $KQ^{-j}\leftrightarrow KR^{-j\pm 4}$) to $KO$-groups, taking
\begin{align}
KO^{-2}(\TT^2)\overset{\sim}{\longleftrightarrow} & KQ^0(\widehat{\TT}^2)\cong \ZZ\oplus\ZZ_2\label{Kgroups2DTI} \\
KO^{-2}(\TT)\overset{\sim}{\longleftrightarrow} & KQ^{-7}(\widehat{\TT})\cong\ZZ_2.\label{Kgroups2DTIdual}
\end{align}

It is instructive to identify generators for the various $K$-theory groups in \eqref{Kgroups2DTI} and \eqref{Kgroups2DTIdual}. First,
\begin{align*}
KO^{-2}(\TT^2) & \cong KO^{-2}(\star)\oplus\widetilde{KO}^{-2}(\TT^2)\nonumber\\
&\cong  KO^{-2}(\star)\oplus\widetilde{KO}^{-2}(S^2)\oplus \widetilde{KO}^{-2}(S^1)\oplus\widetilde{KO}^{-2}(S^1) \nonumber \\
& \cong   \widetilde{KO}^0(S^2)\oplus \widetilde{KO}^0(S^4)\nonumber \\
&\cong \ZZ_2([\mathcal{L}_\mathbb{C}]-[{\bf 1}_\mathbb{C}])\oplus \ZZ([\mathcal{L}_\mathbb{H}]-[{\bf 1}_\mathbb{H}]),\\
KO^{-2}(\TT) & \cong KO^{-2}(\star)\oplus\widetilde{KO}^{-2}(\TT) \cong   \widetilde{KO}^0(S^2)\oplus\widetilde{KO}^0(S^3)\nonumber\\
&\cong \ZZ_2([\mathcal{L}_\mathbb{C}]-[{\bf 1}_\mathbb{C}]),
\end{align*}
where $\mathcal{L}_\mathbb{H}$ is the quaternionic line bundle over $S^4$ associated to the Hopf fibration $\mathrm{Sp}(1)\hookrightarrow S^7\rightarrow S^4$, $\mathcal{L}_\mathbb{C}$ is the complex line bundle associated to $\mathrm{U}(1)\hookrightarrow S^3\rightarrow S^2$, and ${\bf 1}_\mathbb{H}$ and ${\bf 1}_\mathbb{C}$ are respectively the trivial quaternionic and complex line bundles, all regarded as real vector bundles. Furthermore, $\mathcal{L}_\mathbb{H}$ has a non-trivial Pontryagin class $p_1(\mathcal{L}_\mathbb{H})$ generating the top cohomology of $S^4$, while $\mathcal{L}_\mathbb{C}$ has a non-trivial Stiefel--Whitney class $w_2(\mathcal{L}_\mathbb{C})$ in $H^2(S^2,\ZZ_2)$ which can also be identified with $c_1(\mathcal{L}_\mathbb{C})$ modulo 2. We adopt the suggestive notation $\ZZ_2[w_2]$ for the $\ZZ_2$ subgroup of $KO^{-2}(\TT^2)$.

On the right-hand-sides of \eqref{Kgroups2DTI} and \eqref{Kgroups2DTIdual}, the $\ZZ$ factor in $KQ^0(\widehat{\TT}^2)\cong KR^{-4}(\widehat{\TT}^2)$ counts the Quaternionic rank (or half the complex rank), while the $\ZZ_2$ factor is the $K$-theoretic version of the Fu--Kane--Mele invariant. We denote the latter factor by $\ZZ_2[{\rm FKM}]$, which is mapped by $\partial$ onto the group $KQ^{-7}(\widehat{\TT})\cong \ZZ_2$ counting the number of edge Kramers pairs modulo 2. We also write $\ZZ[{\bf 1}_Q]$ for the former $\ZZ$ factor, which is annihilated by $\partial$. We may now deduce that the bulk-boundary homomorphism becomes trivialized on the T-dual side (see also \cite{MT2}),
\beq
\xymatrix{
KO^{-2}(\TT^2)\cong\ZZ[\mathcal{L}_\mathbb{H}-{\bf 1}_\mathbb{H}]\oplus\ZZ_2[w_2] \ar[d]^{\iota^*} \ar[r]^\sim &   \ZZ[{\bf 1}_Q]\oplus\ZZ_2[{\rm FKM}]\cong KQ^0(\widehat{\TT}^2)\ar[d]^\partial \\
KO^{-2}(\TT) \cong \ZZ_2\ar[r]^\sim &  \ZZ_2\cong KQ^{-7}(\widehat{\TT}) 
}\label{Tdualitydiagram2DTI}
\eeq
where $\iota^*$ is the surjective homomorphism induced by the split inclusion $\iota:\TT\hookrightarrow\TT^2$.

\subsection{T-duality for 3D topological insulators}
A $\ZZ_2$-valued invariant of a different nature occurs in 3D topological insulators. As a bulk topological invariant, it lives in $KQ^0(\widehat{\TT}^3)$, which can be computed to be $\ZZ\oplus 4\ZZ_2$ using \eqref{binomialformula}. To see what the generators are more explicitly, we can use \eqref{retractsequence} with $KQ^{p,q}$ instead of $\widetilde{KQ}^{-j}$, and \eqref{retractformula} to obtain
\begin{align}
    KQ^0(\widehat{\TT}^3)&\cong KQ^{0,1}(\widehat{\TT}^2)\oplus KQ^0(\widehat{\TT}^2)\nonumber\\
    &\cong \widetilde{KQ}^{0,1}(\widehat{\TT}^2)\oplus KQ^0(\widehat{\TT}^2)\nonumber \\
    &\cong \widetilde{KQ}^{0,1}(\widehat{S^2})\oplus \widetilde{KQ}^{0,1}(\widehat{S^1})\oplus \widetilde{KQ}^{0,1}(\widehat{S^1})\oplus \left(\widetilde{KQ}^0(\widehat{S^2})\oplus \ZZ[{\bf 1}_Q]\right) \nonumber\\
    &\cong \widetilde{KQ}^{0}(\widehat{S^3}) \oplus 3\widetilde{KQ}^0(\widehat{S^2})\oplus\ZZ[{\bf 1}_Q]. \label{3DKgroups}
\end{align}
The three $\widetilde{KQ}^0(\widehat{S^2})\cong\ZZ_2$ factors are the ``weak'' Fu--Kane--Mele invariants corresponding to 2D sub-tori, while the $\widetilde{KQ}^{0}(\widehat{S^3})\cong \ZZ_2$ factor is a new 3D phenomenon. In the physics literature, this new $\ZZ_2$-invariant was given the interpretation as a topological contribution to the orbital magnetoelectric polarizability $\vartheta\in\{0,\pi\}\subset\mathrm{U}(1)$ 
\cite{QHZ,EMV,FM}, and is sometimes referred to as a \emph{strong} topological invariant. A similar calculation was performed in Section 11 of \cite{FM} in $KR$-theory, where a $\varsigma$-equivariant stable splitting of the $\widehat{\TT}^d$ into a wedge of spheres (c.f.\ \eqref{stablesplitting}) was used. Then the strong $\ZZ_2$ invariant may be identified as the image in $K$-theory under the projection $\widehat{\TT}^3\rightarrow \widehat{S^3}$, while the other three $\ZZ_2$ invariants are the $K$-theory images corresponding to three choices of projections $\widehat{\TT}^3\rightarrow \widehat{S^1}\times\widehat{S^1}\rightarrow\widehat{S^2}$.

The T-duality isomorphism is now
\beq
KO^{-1}(\TT^3) \longleftrightarrow KR^{-4}(\widehat{\TT}^3)\cong KQ^0(\widehat{\TT}^3).\nonumber
\eeq
One way to compute the left-hand-side is to use the fact that $\Sigma(X\times Y)$ is homotopy equivalent to $\Sigma X\vee \Sigma Y\vee\Sigma(X\wedge Y)$, where $\Sigma X=S^1\wedge X$ is the reduced suspension\footnote{The notation is standard and should not be confused with our earlier usage of $\Sigma$ for a Cantor set.} (Proposition 4.I.1 of \cite{Hatcher}). Applying this twice gives
\beq
    \Sigma(S^1\times S^1\times S^1)\simeq S^2\vee S^2\vee S^2\vee S^3\vee S^3\vee S^3\vee S^4.\label{stablesplitting}
\eeq
It follows that 
\begin{align}
KO^{-1}(\TT^3)&=KO^{-1}(\star)\oplus\widetilde{KO}^{-1}(\TT^3)\cong \widetilde{KO}^0(S^1)\oplus \widetilde{KO}^0(\Sigma(\TT^3))\nonumber\\
& \cong \widetilde{KO}^0(S^1)\oplus 3\widetilde{KO}^0(S^2)\oplus \widetilde{KO}^0(S^4)\nonumber\\
&\cong \ZZ_2([\mathcal{L}_\mathbb{R}]-[{\bf 1}_\mathbb{R}])\oplus 3 \ZZ_2([\mathcal{L}_\mathbb{C}]-[{\bf 1}_\mathbb{C}])\oplus \ZZ([\mathcal{L}_\mathbb{H}]-[{\bf 1}_\mathbb{H}]),\label{3DTdualKgroups}
\end{align}
where $\mathcal{L}_\mathbb{R}$ is M\"{o}bius bundle generating $\widetilde{KO}^0(S^1)$. Alternatively, $\mathcal{L}_\mathbb{R}$ is the canonical line bundle over $\mathbb{R}\mathbb{P}^1$ associated to $\mathrm{O}(1)\hookrightarrow S^1\rightarrow S^1$, which has a non-trivial first Stiefel--Whitney class $w_1$ reflecting its non-orientability. 

Comparing \eqref{3DKgroups} and \eqref{3DTdualKgroups}, we see that the four FKM $\ZZ_2$-invariants (or the FKMM invariants in the sense of \cite{deNittis}) for the 3D topological insulator correspond under T-duality to the Stiefel--Whitney classes $w_1, w_2$. Furthermore, the direct sum decomposition is respected \cite{MT2}, so the strong FKM $\ZZ_2$ invariant corresponds to $w_1$, while the weak FKM $\ZZ_2$ invariants correspond to $w_2$.

We can carry out an analysis for the PV-boundary map under the crossed product of one copy of $\ZZ$ in $\ZZ^3$, similar to that done for the 2D case, in order to study the corresponding surface phenomena. The result \cite{MT2} is a commutative diagram
\beq
\xymatrix{
KO^{-1}(\TT^3)\ar[d]^{\iota^*} \ar[r]^\sim &    KQ^0(\widehat{\TT}^3)\ar[d]^\partial \\
KO^{-1}(\TT^2)\ar[r]^\sim &  KQ^{-7}(\widehat{\TT}^2) \nonumber
}.
\eeq

More generally, our notion of bulk-boundary homomorphism makes sense in any $d\geq 1$. There are also T-dualities for twisted (real) group algebras \cite{Rosenberg}, and twisted crossed products can also be considered when disorder needs to be accounted for. Some of these generalizations form the subject of separate works \cite{MT2,HMT}.

\appendix

\section{The twisted foliated index theorem}\label{appendix:indextheorem}

Here we state a special case of the twisted index theorem that we need in this paper, which has been proved in 
\cite{BenameurMathai2015}.

Let $\rho: \Z^2 \longrightarrow {\rm Homeo}(\Sigma)$ denote the minimal action of $\Z^2�$ on $\Sigma$. 
We suppose that $\mu$ is an invariant measure on $\Sigma$ and that $p$ is even. 
Then the suspension $X= \RR^2 \times_{\Z^2}\Sigma$ is a compact foliated space with transversal the Cantor set $\Sigma$,
and with invariant transverse measure induced from $\mu$.

Set $B = \theta dx\wedge dy$ which is a closed 2-form on $\RR^2 \times \Sigma$ (which is supported on $\RR^2$) satisfying $\gamma^*B=B$ for all $\gamma \in \ZZ^2$.
Since $B=d\eta$ where for instance $\eta = \theta x dy$, we get $0=d(\gamma^*\eta - \eta)$. Since $\RR^2$ is simply-connected,
we see that $\gamma^*\eta - \eta= d\phi_\gamma$, where $\phi_\gamma$ is a smooth function on $\RR^2 \times \Sigma$ (which is supported on $\RR^2$). We normalise it 
so that $\phi_\gamma(0)=0$ for all $\gamma\in\ZZ^2$.

Consider functions $f$ in $L^2(\RR^2 \times \Sigma; dxd\mu)$ and bounded operators on it defined as follows,
\begin{enumerate}
\item $S_\gamma f(x,\vartheta) = e^{i\varphi_\gamma(x)} f(x, \vartheta)$;
\item $U_\gamma f(x,\vartheta) = f(x.\gamma, \vartheta.\gamma)$.
\end{enumerate}
Then for all $\gamma \in \ZZ^2$, the bounded operators $T_\gamma=U_\gamma \circ S_\gamma$ satisfy the relation
\beq
T_{\gamma_1} T_{\gamma_1} = \sigma(\gamma_1,\gamma_2)\, T_{\gamma_1\gamma_2}
\eeq
where $ \sigma(\gamma_1,\gamma_2)= \phi_{\gamma_1}(\gamma_2)$ is a multiplier on $\ZZ^2$.

Let $\dirac$ denote the Dirac operator on $\RR^2$ and $\nabla = d+i\eta$ the connection on the trivial line
bundle on $\RR^2$, $\nabla^{\cE}$ the lift to $\RR^2 \times \Sigma$ of a connection on a vector bundle $\cE\to X$ with curvature $F_\cE$. 
Consider the twisted Dirac operator along the leaves of the lifted foliation,
\beq
D= \dirac\otimes\nabla\otimes\nabla^{\cE} : L^2(\RR^2 \times \Sigma, \cS^+\otimes \cE)\longrightarrow L^2(\RR^2 \times \Sigma, \cS^-\otimes \cE).
\eeq
Then one computes that $T_\gamma \circ D= D\circ T_\gamma$ for all $\gamma\in\ZZ^2$.

The {\em twisted foliation analytic index} is a map, generalizing \cite{Marcolli,Mathai99},
$$
{\rm Index}_{C(X)\rtimes_\sigma \RR^2}: K^0(X) \longrightarrow K_0(C(X)\rtimes_\sigma \RR^2).
$$
Since $C(X)\rtimes_\sigma \RR^2$ is strongly Morita equivalent to $C(\Sigma)\rtimes_\sigma \ZZ^2$, the index map is,
$$
{\rm Index}_{C(\Sigma)\rtimes_\sigma \ZZ^2}: K^0(X) \longrightarrow K_0(C(\Sigma)\rtimes_\sigma \ZZ^2).
$$

Let $\tau$ denote the von Neumann trace on $A^\infty_\theta$, which together with $\mu$ induces a trace $\tau_\mu$
on the crossed product $C(\Sigma)\rtimes_\sigma \ZZ^2$. Let $\tau_K$ denote the Kubo conductance cyclic 2-cocycle on $A^\infty_\theta$, which together with $\mu$ induces a
cyclic 2-cocycle on 
$\tau_{K,\mu}$
on a smooth subalgebra $C(\Sigma)\rtimes_\sigma^\infty \ZZ^2$ of $C(\Sigma)\rtimes_\sigma \ZZ^2$ given by a Frechet completion of the 
algebraic crossed product. Then by \cite{BenameurMathai2015},
\beq
\langle [\tau_{K,\mu}], {\rm Index}_{C(\Sigma)\rtimes_\sigma \ZZ^2}(D)\rangle= \frac{1}{(2\pi)^2}  \int_X e^{\theta dx\wedge dy} \wedge {\rm ch}(F_\cE) \wedge dx\wedge dy\, d\mu(\vartheta).
\eeq
where $[\tau_\mu]$ denotes the cyclic cohomology class in $HC^2(C(\Sigma)\rtimes_\sigma^\infty \ZZ^2)$, and $C(\Sigma)\rtimes_\sigma^\infty \ZZ^2$ is a dense subalgebra of $C(\Sigma)\rtimes_\sigma \ZZ^2$ in the domain of the derivations defining $\tau_{K,\mu}$. When the action on $\Sigma$ is minimal, the rank of $\cE$ is constant, and the above pairing reduces to 
\beq
\langle [\tau_{K,\mu}], {\rm Index}_{C(\Sigma)\rtimes_\sigma \ZZ^2}(D)\rangle= \frac{1}{(2\pi)^2}  \int_X {\rm rank}(\cE) \wedge dx\wedge dy\, d\mu(\vartheta)={\rm rank}(\cE).
\eeq

{\em Acknowledgements}  GCT wishes to thank C.\ Bourne, P.\ Bouwknegt, A.\ Carey, Y.\ Kubota, G.\ de Nittis, M.\ Porta, and A.\ Rennie for helpful discussions, as well as the organizers of the Mini-workshop on Topological States and Non-commutative Geometry at WPI-AIMR Tohoku University. The authors also thank D.\ Baraglia for a helpful suggestion. This work was supported by the Australian Research Council via ARC Discovery Project grants DP110100072, DP150100008 and DP130103924.



\begin{thebibliography}{10}

\bibitem{Avron}
Avron, J.E., Pnueli, A.:
Landau Hamiltonians on symmetric spaces. In: Albeverio, S., Fenstad, J.E., Holden, H., Lindstr\o m, T. (eds.) Ideas and methods in quantum and statistical physics, pp. 96--117. Cambridge Univ. Press, Cambridge (1992) 

\bibitem{Avila}
Avila, J.C., Schulz-Baldes, H., Villegas-Blas, C.: Topological invariants of edge states for periodic two-dimensional models. Math. Phys. Anal. Geom. {\bf 16}(2) 137--170 (2013)

\bibitem{Baum}
Baum, P., Karoubi, M.: On the Baum--Connes conjecture in the real case. Quart. J. Math. {\bf 55}(3) 231--235 (2004)

\bibitem{BCH}
Baum, P., Connes, A., Higson, N.: Classifying space for proper actions and $K$-theory of group $C^*$-algebras. Contemp. Math. {\bf 167}  240--291 (1994)

\bibitem{BCL} 
Bellissard, J., Contensous, E., Legrand, A.: {$K$-th\'eorie des quasicristaux, image par la trace: le cas du r\'eseau octogonal}. C. R. Acad. Sci. Sér. I Math. {\bf 326}(2) 197--200 (1998)

\bibitem{Bellissard}
Bellissard, J., van Elst, A., Schulz-Baldes, H.: The noncommutative geometry of the quantum Hall effect. J. Math. Phys. {\bf 35}(10) 5373--5451 (1994)

\bibitem{BO07}
Benameur, M.-T., Oyono-Oyono, H.: Index theory for quasi-crystals I. Computation of the gap-label group. J. Funct. Anal. {\bf 252}(1) 137--170 (2007)

\bibitem{BenameurMathai2015}
Benameur, M.-T. and Mathai, V.: Gap-labelling conjecture with non-zero magnetic field {[{\tt arXiv:1508.01064}]}

\bibitem{BHZ}
Bernevig, B.A., Hughes, T.L., Zhang, S.-C.: Quantum spin Hall effect and topological phase transition in HgTe quantum wells. Science {\bf 314}(5806) 1757--1761 (2006)

\bibitem{Blackadar}
Blackadar, B.: $K$-theory for operator algebras. Math. Sci. Res. Inst. Publ., vol. 5., Cambridge Univ. Press, Cambridge (1998)

\bibitem{Bourne}
Bourne, C., Carey, A.L., Rennie, A.: The Bulk-Edge Correspondence for the Quantum Hall Effect in Kasparov theory. Lett. Math. Phys. {\bf 105}(9) 1253--1273 (2015)

\bibitem{BEM}
Bouwknegt, P.,  Evslin, J., Mathai, V.: T-duality: Topology Change from H-flux. Commun. Math. Phys {\bf 249} 383--415 (2004)
\href{http://arxiv.org/abs/hep-th/0306062}{[{\tt arXiv:hep-th/0306062}]}.

\bibitem{BEM2}
Bouwknegt, P.,  Evslin, J., Mathai, V.: On the Topology and Flux of T-Dual Manifolds. Phys. Rev. Lett. {\bf 92} 181601 (2004)
\href{http://arxiv.org/abs/hep-th/0312052}{[{\tt arXiv:hep-th/0312052}]}.

\bibitem{Brabanter}
De Brabanter, M.: The classification of rational rotation $C^*$-algebras, Arch. Math. {\bf 43}(1) 79--83 (1984)

\bibitem{CHMM}
Carey, A., Hannabuss, K., Mathai, V., McCann, P.: Quantum Hall Effect on the hyperbolic plane. Commun. Math. Phys. {\bf 190}(3) 629-673 (1998)
\href{http://arxiv.org/abs/dg-ga/9704006}{[{\tt arXiv:dg-ga/9704006}]}.

\bibitem{CZ}
Chang C.-Z., et al., Experimental observation of the quantum anomalous Hall effect in a magnetic topological insulator. Science {\bf 340}(6129) 167--170 (2013)

\bibitem{Connes81}
Connes, A.: An analogue of the Thom isomorphism for crossed products of a $C^*$-algebra by an action of $\mathbb{R}$. Adv. Math. {\bf 39}(1) 31--55 (1981)

\bibitem{Connes85}
Connes, A.: Non-commutative differential geometry. Publ. Math. Inst. Hautes \'{E}tude Sci. {\bf 62}(1) 41--144 (1985)

\bibitem{Connes94}
Connes, A.: Noncommutative Geometry. Acad. Press, San Diego (1994)

\bibitem{deNittis}
de Nittis, G., Gomi, K.: Classification of ``Quaternionic'' Bloch-bundles: Topological Insulators of type AII. Commun. Math. Phys. {\bf 339}(1) 1--55 (2015)

\bibitem{Dupont}
Dupont, J.L.: Symplectic Bundles and $KR$-Theory. Math. Scand. {\bf 24} 27--30 (1969)

\bibitem{Elbau}
Elbau, P. Graf, G.M.: Equality of bulk and edge Hall conductance revisited. Commun. Math. Phys. {\bf 229}(3) 415--432(2002)

\bibitem{EMV}
Essin, A.M., Moore, J.E., Vanderbilt, D.: Magnetoelectric polarizability and axion electrodynamics in crystalline insulators. Phys. Rev. Lett. {\bf 102} 146805 (2009)

\bibitem{FM}
Freed, D.S., Moore, G. W.: Twisted equivariant matter.
Ann. Henri Poincar{\'e} {\bf 14} 1927--2023 (2013) 

\bibitem{FK2}
Fu, L., Kane, C.L.: Time reversal polarization and a $\mathbb{Z}_2$ adiabatic spin pump. Phys. Rev. B {\bf 74}(19) 195312 (2006)

\bibitem{FKM}
Fu, L., Kane, C.L., Mele, E.J.: Topological insulators in three dimensions.
Phys. Rev. Lett. {\bf 98}(10) 106803 (2007)

\bibitem{FKMM}
Furuta, M., Kametani, Y., Matsue, H., Minami, N.: Stable-homotopy Seiberg-Witten invariants and Pin bordisms. UTMS Preprint Series 2000, UTMS 2000-46 (2000)

\bibitem{Graf}
Graf, G.M., Porta, M.: Bulk-edge correspondence for two-dimensional topological insulators. Commun. Math. Phys. {\bf 324}(3) 851--895 (2013)

\bibitem{Green}
Green, P.: The local structure of twisted covariance algebras. Acta. Math. {\bf 140}(1) 191--250 (1978)

\bibitem{Hal}
Haldane, F.D.M.: Model for a quantum Hall effect without Landau levels: Condensed-matter realization of the parity anomaly. Phys. Rev. Lett. {\bf 61}(18) 2015 (1988)

\bibitem{HMT}
Hannabuss, K.C., Mathai, V., Thiang, G.C.: T-duality trivializes bulk-boundary correspondence: the parametrised case. {[{\tt arXiv:1510.04785}]}

\bibitem{Hatcher}
Hatcher, A.: Algebraic topology. Cambridge Univ. Press, Cambridge (2002)

\bibitem{Hatsugai}
Hatsugai, Y.: Chern number and edge states in the integer quantum Hall effect. Phys. Rev. Lett. {\bf 71}(22) 3697 (1993)

\bibitem{H} 
Hori, K.: D-branes, T-duality, and index theory. Adv. Theor. Math. Phys. {\bf 3} 281--342 (1999)


\bibitem{Hsieh}
Hsieh, D., Qian, D., Wray, L., Xia, Y., Hor, Y. S., Cava, R.J., Hasan, M.Z.: A topological Dirac insulator in a quantum spin Hall phase. Nature  {\bf 452}(7190) 970--974 (2008)

\bibitem{JM}
Jotzu, G. Messer, M., Desbuquois, R., Lebrat, M., Uehlinger, T., Greif, D., Esslinger, T.: Experimental realization of the topological Haldane model with ultracold fermions.
Nature {\bf 515}(7526) 237--240 (2014)

\bibitem{KM} 
Kane, C.L., Mele, E.J.: Quantum Spin Hall Effect in Graphene.
Phys. Rev. Lett. {\bf 95}(22) 226801 (2005) 

\bibitem{KM2}
Kane, C.L., Mele, E.J.: $\mathbb{Z}_2$ Topological Order and the Quantum Spin Hall Effect.
Phys. Rev. Lett. {\bf 95}(14) 146802 (2005)


\bibitem{Kotani}
Kotani, M., Schulz-Baldes, H., Villegas-Blas, C.: Quantization of interface currents. J. Math. Phys. {\bf 55}(12) 121901 (2014)

\bibitem{Kellendonk4}
Kellendonk, J., Richard, S. Topological boundary maps in physics: General theory and applications. {[\tt arXiv:math-ph/0605048]}

\bibitem{Kellendonk1}
Kellendonk, J., Richter, T., Schulz-Baldes, H.: Edge current channels and Chern numbers in the integer quantum Hall effect. Rev. Math. Phys. {\bf 14}(1) 87--119 (2002)

\bibitem{Kellendonk2}
Kellendonk, J., Schulz-Baldes, H.: Quantization of edge currents for continuous magnetic operators. J. Funct. Anal. {\bf 209}(2) 388--413 (2004)

\bibitem{Kellendonk3}
Kellendonk, J., Schulz-Baldes, H.: Boundary Maps for $C^*$-Crossed Products with with an Application to the Quantum Hall Effect. Commun. Math. Phys. {\bf 249}(3) 611--637 (2004)

\bibitem{KW}
K\"{o}nig, M., Wiedmann, S., Br\"{u}ne, C., Roth, A., Buhmann, H., Molenkamp, L.W., Qi, X.-L., Zhang, S.-C.: Quantum spin Hall insulator state in HgTe quantum wells. Science {\bf 318}(5851) 766--770 (2007)

\bibitem{AK}
Kitaev, A.: Periodic table for topological insulators and superconductors. In: AIP Conf. Proc., vol. 1134, no. 1, pp.\ 22--30 (2009)

\bibitem{Loring}
Loring, T.A.: $K$-theory and pseudospectra for topological insulators. Ann. Physics {\bf 356} 383--416 (2015)

\bibitem{Luke}
Luke, G., Mishchenko. A.S.: Vector bundles and their applications. Kluwer, Boston (1998)

\bibitem{Marcolli}
Marcolli, M., Mathai, V.: Twisted index theory on good orbifolds. II. Fractional quantum numbers. Commun. Math. Phys. {\bf 217}(1) 55--87 (2001)
\href{http://arxiv.org/abs/math/9911103}{[{\tt arXiv:math/9911103}]}

\bibitem{Mathai99}
Mathai, V.: $K$-theory of twisted group $C^*$-algebras and positive scalar curvature. 
Contemp. Math. {\bf 231} 203--225 (1999)

\bibitem{MR05}
Mathai, V., Rosenberg, J.: T-duality for torus bundles with H-fluxes via noncommutative topology. Commun. Math. Phys. {\bf 253}(3) 705--721 (2005)
\href{http://arxiv.org/abs/hep-th/0401168}{[{\tt arXiv:hep-th/0401168}]}

\bibitem{MR06}
Mathai, V., Rosenberg, J.: T-duality for torus bundles with H-fluxes via noncommutative topology, II;
the high-dimensional case and the T-duality group. Adv. Theor. Math. Phys. {\bf 10} 123-158 (2006)  
\href{http://arxiv.org/abs/hep-th/0508084}{[{\tt arXiv:hep-th/0508084}]}.

\bibitem{MT}
Mathai, V., Thiang, G.C.: T-duality and topological insulators. 
J. Phys. A: Math. Theor. (Fast Track Communications) {\bf 48}(42) 42FT02 (2015)
\href{http://arxiv.org/abs/1503.01206}{[{\tt arXiv:1503.01206}]}


\bibitem{MT2}
Mathai, V., Thiang, G.C.: 
T-duality trivializes bulk-boundary correspondence: some higher dimensional cases, 15pp \href{http://arxiv.org/abs/1506.04492 }{[{\tt arXiv:1506.04492} ]}

\bibitem{Nest}
Nest, R.: Cyclic cohomology of crossed products with $\mathbb{Z}$. J. Funct. Anal. {\bf 80}(2) 235--283 (1988)

\bibitem{Nistor}
Nistor, V.: Higher index theorems and the boundary map in cyclic cohomology. Doc. Math. {\bf 2} 263--295 (1997)

\bibitem{PR} 
Packer, J., Raeburn, I.: Twisted crossed products of $C^*$-algebras. Math. Proc. Cambridge Philos. Soc. {\bf 106}(2) 293--311 (1989)

\bibitem{Pimsner}
Pimsner, M., Voiculescu, D.: Exact sequences for $K$-groups and $EXT$-groups of certain cross-product $C^*$-algebras. J. Operator Theory {\bf 4} 93--118 (1980)

\bibitem{Prodan}
Prodan, E.: Virtual Topological Insulators with Real Quantized Physics. {[\tt arXiv:1503.04757]}

\bibitem{Prodan2}
Prodan, E.: Robustness of the spin-Chern number. Phys. Rev. B {\bf 80}(12) 125327 (2009)

\bibitem{Putnam}
Putnam, I.F.: The $C^*$-algebras associated with minimal homeomorphisms of the Cantor set. Pacific J. Math. {\bf 136}(2) 329--353 (1989)

\bibitem{QHZ}
Qi, X.-L., Hughes, T.L., Zhang, S.-C.: Topological field theory of time-reversal invariant insulators. Phys. Rev. B {\bf 78} 195424 (2008)

\bibitem{RS}
Reed, M., Simon, B.: Methods of Mathematical Physics Vol 4: Analysis of Operators. Academic Press, New York (1978)

\bibitem{Rieffel}
Rieffel, M.A.: $C^*$-algebras associated with irrational rotations. Pacific J. Math. {\bf 93}(2) 415--429 (1981)

\bibitem{Rieffel82}
Rieffel, M.A.: 
Applications of strong Morita equivalence to transformation group $C^*$-algebras. Operator algebras and applications, Part I (Kingston, Ont., 1980), pp. 299--310, 
Proc. Sympos. Pure Math., {\bf 38}, Amer. Math. Soc., Providence, R.I. (1982) 

\bibitem{Rosenberg}
Rosenberg, J.: Real Baum--Connes assembly and T-duality for torus orientifolds. J. Geom. Phys. {\bf 89} 24--31 (2015)

\bibitem{Rosenberg2}
Rosenberg, J.: $C^*$-algebras, positive scalar curvature, and the Novikov Conjecture. Publ. Math. Inst. Hautes \'{E}tude Sci. {\bf 58}(1) 197--212 (1983)

\bibitem{Savinien}
Savinien, J., Bellissard, J.: A spectral sequence for the $K$-theory of tiling spaces. Ergodic Theory Dynam. Systems {\bf 29} 997--1031 (2009)

\bibitem{Schroder}
Schr\"{o}der, H.: $K$-theory for real $C^*$-algebras and applications. Pitman Res. Notes Math. Ser. vol. 290 (1993)


\bibitem{Sheng}
Sheng, D.N., Weng, Z.Y., Sheng, L., Haldane, F.D.M.: Quantum spin-Hall effect and topologically invariant Chern numbers. Phys. Rev. Lett. {\bf 97}(3) 036808 (2006)

\bibitem{SPF}
Sticlet, D., Pi\'{e}chon, F., Fuchs, J.-N., Kalugin, P., Simon, P.: Geometrical engineering of a two-band Chern insulator in two dimensions with arbitrary topological index. Phys. Rev. B {\bf 85}(16) 165456 (2012)

\bibitem{T}
Thiang, G.C.: On the $K$-theoretic classification of topological phases of matter. Ann. Henri Poincar{\'e} \href{http://link.springer.com/article/10.1007/s00023-015-0418-9}{(Online First)} \href{http://arxiv.org/abs/1406.7366}{[{\tt arXiv:1406.7366}]}

\bibitem{T2}
Thiang, G.C.: Topological phases: homotopy, isomorphism and $K$-theory. Int. J. Geom. Methods Mod. Phys. {\bf 12}(9) 1550098 (2015)

\bibitem{Williams}
Williams, D.P.: Crossed products of $C^*$-algebras. Math. Surveys Monogr., vol. 134, Amer. Math. Soc., Providence (2007)


\end{thebibliography}
\end{document}